\PassOptionsToPackage{unicode}{hyperref}
\PassOptionsToPackage{hyphens}{url}
\PassOptionsToPackage{dvipsnames,svgnames,x11names}{xcolor}
\documentclass[
  12pt]{article}

\usepackage{amsmath,amssymb}
\usepackage{iftex}
\ifPDFTeX
  \usepackage[T1]{fontenc}
  \usepackage[utf8]{inputenc}
  \usepackage{textcomp} % provide euro and other symbols
\else % if luatex or xetex
  \usepackage{unicode-math}
  \defaultfontfeatures{Scale=MatchLowercase}
  \defaultfontfeatures[\rmfamily]{Ligatures=TeX,Scale=1}
\fi
\usepackage{lmodern}
\ifPDFTeX\else  
\fi
\IfFileExists{upquote.sty}{\usepackage{upquote}}{}
\IfFileExists{microtype.sty}{% use microtype if available
  \usepackage[]{microtype}
  \UseMicrotypeSet[protrusion]{basicmath} % disable protrusion for tt fonts
}{}
\makeatletter
\@ifundefined{KOMAClassName}{% if non-KOMA class
  \IfFileExists{parskip.sty}{%
    \usepackage{parskip}
  }{% else
    \setlength{\parindent}{0pt}
    \setlength{\parskip}{6pt plus 2pt minus 1pt}}
}{% if KOMA class
  \KOMAoptions{parskip=half}}
\makeatother
\usepackage{xcolor}
\makeatletter
\ifx\paragraph\undefined\else
  \let\oldparagraph\paragraph
  \renewcommand{\paragraph}{
    \@ifstar
      \xxxParagraphStar
      \xxxParagraphNoStar
  }
  \newcommand{\xxxParagraphStar}[1]{\oldparagraph*{#1}\mbox{}}
  \newcommand{\xxxParagraphNoStar}[1]{\oldparagraph{#1}\mbox{}}
\fi
\ifx\subparagraph\undefined\else
  \let\oldsubparagraph\subparagraph
  \renewcommand{\subparagraph}{
    \@ifstar
      \xxxSubParagraphStar
      \xxxSubParagraphNoStar
  }
  \newcommand{\xxxSubParagraphStar}[1]{\oldsubparagraph*{#1}\mbox{}}
  \newcommand{\xxxSubParagraphNoStar}[1]{\oldsubparagraph{#1}\mbox{}}
\fi
\makeatother

\usepackage{longtable,booktabs,array}
\usepackage{calc} % for calculating minipage widths
\usepackage{etoolbox}
\makeatletter
\patchcmd\longtable{\par}{\if@noskipsec\mbox{}\fi\par}{}{}
\makeatother
\IfFileExists{footnotehyper.sty}{\usepackage{footnotehyper}}{\usepackage{footnote}}
\makesavenoteenv{longtable}
\usepackage{graphicx}
\makeatletter
\def\maxwidth{\ifdim\Gin@nat@width>\linewidth\linewidth\else\Gin@nat@width\fi}
\def\maxheight{\ifdim\Gin@nat@height>\textheight\textheight\else\Gin@nat@height\fi}
\makeatother
\setkeys{Gin}{width=\maxwidth,height=\maxheight,keepaspectratio}
\makeatletter
\def\fps@figure{htbp}
\makeatother

\makeatletter
\@ifpackageloaded{caption}{}{\usepackage{caption}}
\AtBeginDocument{%
\ifdefined\contentsname
  \renewcommand*\contentsname{Table of contents}
\else
  \newcommand\contentsname{Table of contents}
\fi
\ifdefined\listfigurename
  \renewcommand*\listfigurename{List of Figures}
\else
  \newcommand\listfigurename{List of Figures}
\fi
\ifdefined\listtablename
  \renewcommand*\listtablename{List of Tables}
\else
  \newcommand\listtablename{List of Tables}
\fi
\ifdefined\figurename
  \renewcommand*\figurename{Figure}
\else
  \newcommand\figurename{Figure}
\fi
\ifdefined\tablename
  \renewcommand*\tablename{Table}
\else
  \newcommand\tablename{Table}
\fi
}
\@ifpackageloaded{float}{}{\usepackage{float}}
\floatstyle{ruled}
\@ifundefined{c@chapter}{\newfloat{codelisting}{h}{lop}}{\newfloat{codelisting}{h}{lop}[chapter]}
\floatname{codelisting}{Listing}

\makeatother
\makeatletter
\@ifpackageloaded{caption}{}{\usepackage{caption}}
\@ifpackageloaded{subcaption}{}{\usepackage{subcaption}}
\makeatother

\ifLuaTeX
  \usepackage{selnolig}  % disable illegal ligatures
\fi

\usepackage{bookmark}

\usepackage{amsmath}
\usepackage{graphicx,psfrag,epsf}
\usepackage{enumerate}
\usepackage{natbib}
\usepackage{array}
\usepackage{appendix}
\usepackage{url} % not crucial - just used below for the URL 
\usepackage{tabu}
\usepackage{ulem}
\usepackage{dsfont}
\usepackage{bm}
\usepackage{soul}
\usepackage{verbatim}
\usepackage{amsfonts}
\usepackage[english]{babel}
\usepackage{subcaption}
\usepackage[utf8]{inputenc}
\usepackage{booktabs}
\usepackage[T1]{fontenc}
\usepackage{amsthm}
\usepackage{thmtools}
\usepackage{xcolor} 
\usepackage{mathtools}
\usepackage{bm}
\usepackage{subfiles}
\usepackage[export]{adjustbox}
\usepackage{float}
\usepackage{tablefootnote}
\usepackage{cleveref}
\usepackage{booktabs}
\usepackage{multirow}
\usepackage{listings}
\usepackage{xcolor}

\usepackage{tikz}
\usetikzlibrary{arrows.meta}
\newtheorem{proposition}{Proposition}[section]

\newtheorem*{remark}{Remark}
\newtheorem{lemma}{Lemma}[section] 
\newtheorem{definition}{Definition}[section]

\definecolor{danubeBlue}{RGB}{50,116,161}
\newcommand{\TrainStartDate}{2024-02-28}
\newcommand{\TrainEndDate}{2025-02-28}
\newcommand{\NumberOfProduct}{$69$}
\newcommand{\NumObsRiver}{$4,692$}
\newcommand{\NumObsChina}{$9149$}

\IfFileExists{xurl.sty}{\usepackage{xurl}}{} % add URL line breaks if available
\hypersetup{
  pdftitle={Title},
  pdfauthor={Author 1; Author 2},
  pdfkeywords={3 to 6 keywords, that do not appear in the title},
  colorlinks=true,
  linkcolor={blue},
  filecolor={Maroon},
  citecolor={Blue},
  urlcolor={Blue},
  pdfcreator={LaTeX via pandoc}}

\newcommand{\anon}{1}

\begin{document}

\def\spacingset#1{\renewcommand{\baselinestretch}%
{#1}\small\normalsize} \spacingset{1}

%%%%%%%%%%%%%%%%%%%%%%%%%%%%%%%%%%%%%%%%%%%%%%%%%%%%%%%%%%%%%%%%%%%%%%%%%%%%%%

\ifnum\anon=1
  \title{\bf Separation-based causal discovery for extremes}

  \author{
    Junshu Jiang$^{\rm a,*}$\hspace{.2cm},
    Jordan Richards$^{\rm b}$,
    Rapha\"el Huser$^{\rm a}$,
    David Bolin$^{\rm a}$\\
    $\;^{\rm a}$ Statistics Program, CEMSE Division,\\
    King Abdullah University of Science and Technology, Saudi Arabia.\\
    $\;^{\rm b}$ School of Mathematics and Maxwell Institute for Mathematical Sciences,\\
    University of Edinburgh, Edinburgh, UK.
  }

  \maketitle

  \begingroup
  \renewcommand\thefootnote{*}
  \footnotetext{Corresponding author: junshu.jiang@kaust.edu.sa.}
  \endgroup

  \begingroup
  \renewcommand\thefootnote{}
  \footnotetext{The authors gratefully acknowledge funding from the King Abdullah University of Science and Technology (KAUST) Office of Sponsored Research (OSR) under Award No.\ OSR-CRG2020-4394 and ORFS-CRG12-2024-6399.}
  \endgroup
\fi
\bigskip
\begin{abstract}
Structural causal models (SCMs), with an underlying directed acyclic graph (DAG), provide a powerful analytical framework to describe the interaction mechanisms in large-scale complex systems. However, when the system exhibits extreme events, the governing mechanisms can change dramatically, and SCMs with a focus on rare events are needed. We propose a new class of SCMs, called XSCMs, which leverage transformed-linear algebra to model causal relationships among extreme values. Similar to traditional SCMs, we prove that XSCMs satisfy the causal Markov and causal faithfulness properties with respect to partial tail (un)correlatedness. This enables estimation of the underlying DAG for extremes using separation-based tests, and makes many state-of-the-art constraint-based causal discovery algorithms directly applicable. We further consider the problem of undirected graph estimation for relationships among tail-dependent (and potentially heavy-tailed) data. The effectiveness of our method, compared to alternative approaches, is validated through simulation studies on large-scale systems with up to 50 variables, and in a well-studied application to river discharge data from the Danube basin. Finally, we apply the framework to investigate complex market-wide relationships in China's derivatives market.
\end{abstract}

\noindent%
\vfill

\newpage
\spacingset{1.8} % DON'T change the spacing!

\section{Introduction}\label{sec:intro}
A primary goal of statistical learning is to characterize dependencies among random variables. Beyond capturing statistical associations, a key consideration is to uncover causal relationships, as these provide key insights into counterfactual questions and evaluation of intervention effects~\citep[see, e.g.,][]{hammoudeh2020relationship,shojaie2022granger}. While statistical associations alone do not imply causation, causal discovery from observational data becomes possible under strong assumptions, such as causal Markovianity, causal faithfulness, and causal sufficiency, which connect the data-generating process to testable statistical properties~\citep[][Section~2]{pearl2009causality}. By contrast with designed experiments crafted specifically to answer causal questions of interest~\citep[see, e.g.,][]{imai2013experimental}, the development of statistical approaches that can infer causal relationships from observational data is especially valuable when it is prohibitively costly, technically challenging, or ethically questionable to conduct such designed experiments. This includes %important 
applications in policy making, financial regulation, health-related research, ecological modeling, among others.

Structural causal models (SCMs), also known as structural equation models, assume that there is an induced directed acyclic graph (DAG) driving a set of structural equations that determine the system's behavior. This underlying DAG provides a (typically sparse) graphical representation of causal interactions among the random variables under consideration. Compared to undirected graphical models, SCMs are more informative, using directions of edges to convey causation or directional influence between variables.

SCMs designed for extreme events are particularly important. In some cases, causal relationships may only manifest during extreme events, or may differ significantly between extreme and normal conditions. A typical example is in modern portfolio management, which often relies on holding uncorrelated assets to reduce risk. During market crashes, however, the dependence strength among assets can increase dramatically, undermining portfolio diversification. Despite its importance, causal modeling and discovery for extreme events is a relatively new research direction that has gained increasing attention over the last five years \citep[see, e.g.,][]{gnecco2021causal,pasche_causal_2023,bodik2024causality,bodik2024granger}. In their seminal work, \cite{gnecco2021causal} introduced a linear SCM for heavy-tailed random variables and proposed using the so-called causal tail coefficient as a measure of causal influence between variables. The authors further developed an algorithm to estimate the partial causal order by iteratively computing the causal tail coefficient between pairs of variables. \cite{pasche_causal_2023} extended this approach by considering a parametric form based on the generalized Pareto distribution, which allows estimating the conditional causal tail coefficient while accounting for unobserved confounders. \cite{bodik2024causality} later adapted the causal tail coefficient to the time series context while~\cite{bodik2024granger} developed an algorithm to estimate the underlying DAG utilizing the tail causal coefficient for time series data. Furthermore, \cite{paluvs2024causes} introduced an information-theoretic framework to causal discovery in extremes based on Rényi information transfer---an approach based on a generalization of Shannon's entropy that uses a tunable parameter to characterize the contribution of rare events. \cite{kluppelberg2021estimating} developed a recursive max-linear representation for extremal dependence and proposed a top-down approach to estimating the underlying DAG structure. Later, \cite{krali2025causal} adapted this top-down approach for causal discovery. Recently, \cite{engelke2025extremes} studied the asymptotic behavior of general structural causal models under extreme events. They showed that the extremal dependence structure induced by an SCM can differ from the original causal graph and may correspond to a pruned version of it. Based on this insight, they proposed a pruning procedure to recover aspects of the extremal causal structure. A broader literature review of the intersection between extreme statistics and causal discovery can be found in \cite{chavez2024causality}.

A closely related research area is that of graphical modeling for extremal dependence~\citep{wan2023graphical,tran2024estimating,lederer2023extremes,engelke2020graphical,engelke2021learning,engelke2021sparse,kluppelberg2021estimating,engelke2022structure,krali2023heavy,gong2024partial}. Similarly to SCMs, graphical models for extremal dependence also provide a graphical representation of variable interactions, and some DAG-based models have the potential to be extended to capture extremal causal relationships \citep{krali2023heavy}. \cite{engelke2020graphical} introduced the concept of extremal conditional independence within the framework of multivariate Pareto distributions and proposed a method for estimating undirected graph structures in H\"usler--Reiss models of extremal dependence. A comprehensive review of graphical models for extremal dependence is provided in~\cite{engelke2024graphical}.

Most extremal causal modeling approaches rely on the causal tail coefficient introduced by~\cite{gnecco2021causal}. However, this metric is intrinsically bivariate and not sufficient to ensure causal Markovianity and causal faithfulness of classical SCMs, thus failing to connect separation in the underlying DAG with testable statistical properties. In this paper, we instead propose a new SCM for tail-dependent (and potentially heavy-tailed) extremes, the XSCM, that incorporates the transformed linear algebra framework for heavy-tailed random variables introduced by~\cite{cooley2019decompositions}. We propose using the partial tail-correlation coefficient introduced by~\cite{gong2024partial} and~\cite{kim2022hypothesis} to quantify causal relationships between heavy-tailed variables. We show that, when endowed with this metric, the XSCM satisfies the so-called ``tail-causal Markov'' and ``tail-causal faithfulness'' conditions, which are analogous to the causal Markov and faithfulness conditions of classical SCMs. These two conditions imply an equivalence between separation in the induced DAG and partial tail uncorrelation, which guarantees that existing constraint-based causal discovery algorithms can be directly adapted to the extremes context, when applied to both cross-sectional and time series data. Additionally, a by-product of our proposed method is that we can simultaneously estimate the undirected graph structure of the models described in~\cite{gong2024partial} and~\cite{kim2022hypothesis}. The effectiveness of our method compared to existing approaches is demonstrated through comparative simulations on large-scale systems, and in an application to river discharge data from the Danube basin. We also investigate market-wide causal relationships in China's futures market for extremely large trading activities. Here, the estimated DAG aligns with actual market sector categorizations, and the edge directions provide insights into key information flow.

The remainder of this paper is organized as follows. Section~\ref{sec:bg} provides background on SCMs and relevant multivariate extreme value theory. Section~\ref{sec:extremal_causality} introduces our proposed XSCM for modeling extremal events and capturing their causal interactions. Section~\ref{sec:simulation} presents simulation studies for models defined with both directed and undirected graphs. Section~\ref{sec:app} presents the two data applications, and Section~\ref{sec:conclusion} concludes with a discussion on future research directions. Proofs and further details are provided in Appendices, while supporting code and data are available at \if1\anon{\url{https://github.com/junshujiang/tailSCM}}\else{\url{https://github.com/anonymous/anonymous}}\fi.

\section{Background}
\label{sec:bg}

In Section~\ref{sec:causal:causal_graph_model}, we first introduce some commonly used graph notation and then provide preliminaries on structural causal models (SCMs), while Section~\ref{sec:bg:RegularVarying} presents the partial tail-correlation coefficient for regularly varying random vectors.

\subsection{Graph notation and structural causal models}
\label{sec:causal:causal_graph_model}
%\label{sec:graph:notation}
A graph $\mathcal{G} := (\mathcal{V}, \mathcal{E})$ consists of a set of $p$ vertices $\mathcal{V} := \{v_1, \dots, v_p\}$ and a set of edges $\mathcal{E} := \{(v_i, v_j) \in \mathcal{V} \times \mathcal{V}\}$. An undirected graph only has bidirectional edges, i.e., $(v_i, v_j) \in \mathcal{E}$ if and only if $(v_j, v_i) \in \mathcal{E}$. A directed graph, in contrast, contains edges with a specified direction, i.e., one can have $(v_i, v_j) \in \mathcal{E}$ and $(v_j, v_i) \notin \mathcal{E}$. In this case, we write $v_i \rightarrow v_j$ to indicate that there is a directed edge from $v_i$ to $v_j$.

A sequence of vertices $(v_{k_0}, v_{k_1}, \dots, v_{k_n})$ is called a \textit{path} from $v_{k_0}$ to $v_{k_n}$ if each consecutive pair is connected by an edge in either direction, i.e., $(v_{k_i}, v_{k_{i+1}}) \in \mathcal{E}$ or $(v_{k_{i+1}}, v_{k_i}) \in \mathcal{E}$ for $i = 0, \dots, n-1$. It is called a \textit{directed path} if $(v_{k_i}, v_{k_{i+1}}) \in \mathcal{E}$ for $i = 0, \dots, n-1$. A directed graph is said to be \textit{acyclic} if there exists no directed path that starts and ends at the same vertex. 
In a directed acyclic graph (DAG), vertex $v_j$ is called a \textit{parent} of $v_i$ if $(v_j, v_i) \in \mathcal{E}$, and $v_i$ is a \textit{descendant} of $v_j$ if there exists a directed path from $v_j$ to $v_i$.

%\subsection{}

We now review background on structural causal models (SCMs).

\begin{definition}[Structural causal model (\citealp{peters2017elements}; Definition~6.2)]\label{def:SCM}
    
A structural causal model (SCM) for a $p$-dimensional random vector $\bm{X}:=(X_1,\dots,X_p)^\top$ is defined as the pair $\mathcal{S}:=(\bm{Z},\mathcal{F})$ where $\bm{Z}=(Z_1,\dots,Z_p)^\top$ are source random variables and $\mathcal{F}$ is a collection of $p$ structural equations of the form $X_i=f_i({\rm Pa}(X_{i}),Z_i)$ for each $X_i$, with ${\rm Pa}(X_{i})\subset \{X_1,\dots,X_p\}\setminus \{X_i\}$ denoting the \textit{causal parents} of $X_i$, and $f_i$ being measurable functions, for all $i=1\dots,p$.
\end{definition}
%refers to random variables appearing in the right-hand side of the equation $X_i=f_i({\rm Pa}(X_{i}),Z_i)$ for $i=1,\dots,p$.
The underlying directed graph $\mathcal{G}:=(\mathcal{V},\mathcal{E})$ induced by an SCM $\mathcal{S}$ is assumed to be acyclic throughout this work; its vertices are given by $\mathcal{V}:=\{v_1,\dots,v_p\}$ and directed edges by $\mathcal{E}:=\{(v_j,v_i)\in \mathcal{V}\times \mathcal{V}:X_j \in {\rm Pa}(X_{i})\}$. As an illustrative example, consider a simple SCM for the random vector $\bm{X}=(X_1,X_2)^\top$ with Gaussian source variables ${\bm{Z}=(Z_1,Z_2)^\top\sim \mathcal{N}_2(\bm{0}, \bm{I})}$ and structural equations ${X_1=Z_1}$ and ${X_2=X_1+Z_2}$.
The induced graph $\mathcal{G}$ consists of the vertex set $\mathcal{V}=\{v_1,v_2\}$ and the edge set ${\mathcal{E}=\{(v_1,v_2)\}}$. The causal parents of $X_1$ is the empty set, ${\rm Pa}(X_1)=\emptyset$, and the causal parents of $X_2$ are ${\rm Pa}(X_2)=\{X_1\}$. Thus, $\mathcal{G}$ has only one edge, $v_1\rightarrow v_2$.

Statistical inference for graphs involves inferring $\mathcal{G}$ based on observations of $\bm{X}$. Constraint-based methods \citep[see, e.g.,][]{spirtes2001anytime,runge2015identifying} typically convert this task into the problem of detecting all separation relations in the graph. For DAGs, the notion of d-separation is used to represent such separation relations.

\begin{definition}[Collider]\label{def:collider}
    Given three connected vertices $v_i,v_j,v_k \in \mathcal{V}$, the vertex $v_k$ is called a \textit{collider} of $v_i$ and $v_j$ if both arrows are incoming, i.e., $v_i\rightarrow v_k \leftarrow v_j$.
\end{definition}

\begin{definition}[d-separation~(\citealp{pearl2014probabilistic}; Sec. 3.3.1)]\label{def:d_sep}
    Given a DAG $\mathcal{G}:=(\mathcal{V},\mathcal{E})$ with vertex set $\mathcal{V}=\{v_1,\dots,v_p\}$, for any two vertices $v_i,v_j\in \mathcal{V}$ and a separation set ${\bm{\mathcal{S}}_v\subset \mathcal{V}\setminus\{v_i,v_j\} }$, we say that $v_i$ and $v_j$ are d-separated given $\bm{\mathcal{S}}_v$ in the graph $\mathcal{G}$, written as $v_i \perp_{\mathcal{G}} v_j\mid\bm{\mathcal{S}}_v$, if every path between $v_i$ and $v_j$ is blocked by $\bm{\mathcal{S}}_v$. A path is blocked if either: i) it contains a \textit{collider} within the path and neither the collider itself nor its descendant are in $\bm{\mathcal{S}}_v$; or ii) it contains a \textit{non-collider} that is in $\bm{\mathcal{S}}_v$.
\end{definition}

Two desirable properties of classical Gaussian-based SCMs are the \textit{causal Markov} and \textit{causal faithfulness} conditions, which connect d-separation in the underlying DAG and conditional independencies in the distribution of $\bm{X}$. For an SCM over the random vector $\bm{X}=(X_1,\dots,X_p)^\top$ with induced graph $\mathcal{G}=(\mathcal{V},\mathcal{E})$, the causal Markov condition states that, for any two vertices $v_i,v_j\in \mathcal{V}$, and any separation set $\bm{\mathcal{S}}_v\subset \mathcal{V}\setminus\{v_i,v_j\}$ such that  $v_i \perp_{\mathcal{G}} v_j\mid\bm{\mathcal{S}}_v$, variables $X_i$ and $X_j$ are conditionally independent given $\bm{\mathcal{S}}:=\{X_k\}_{v_k\in \bm{\mathcal{S}}_v}$. The causal faithfulness condition states that the reverse implication holds, meaning that conditional independence between $X_i$ and $X_j$ given $\bm{\mathcal{S}}$ implies d-separation, $v_i \perp_{\mathcal{G}} v_j\mid\bm{\mathcal{S}}_v$. 

When both conditions hold, the skeleton of the underlying DAG of an SCM can be estimated by identifying all triplets $(v_i,v_j,\bm{\mathcal{S}}_v)$ such that $v_i\perp_{\mathcal{G}} v_j\mid\bm{\mathcal{S}}_v$.  Note, however, that d-separation information does not uniquely determine the direction of edges in the graph $\mathcal{G}$. Two DAGs, $\mathcal{G}_1$ and $\mathcal{G}_2$, are said to be \textit{Markov equivalent} if they encode the same d-separation relations. The possibility of Markov equivalence causes difficulties in estimating directed graphs. For example, the two DAGs ${v_1\rightarrow v_2 \rightarrow v_3}$ and $v_1\leftarrow v_2 \leftarrow v_3$ are Markov equivalent, and cannot be distinguished by d-separation information.

SCMs can be extended from random vectors to multivariate time series~\citep[see e.g.,][]{runge2018causal}. The structural equations for a $p$-variate time series, $\{X_{i,t}\}_{t\in \mathcal{T}}$ with ${i=1,\dots,p}$ and time index set $\mathcal{T}$, are defined as $X_{i,t}:=f_{i,t}\{{\rm Pa}(X_{i,t}),Z_{i,t}\}$ for $t\in\mathcal{T}$ and  $\{Z_{i,t}\}_{i=1,\dots,p,t\in\mathcal{T}}$ are the source variables. The underlying graph is a DAG with vertex set ${\mathcal{V}:=\{v_{i,t}\}_{i=1,\dots,p;~t\in\mathcal{T}}}$ and edge set ${\mathcal{E}:=\{(v_{j,t^\prime},v_{i,t})\in \mathcal{V}\times \mathcal{V}:X_{j,t^\prime} \in {\rm Pa}(X_{i,t})\}}$ for $i,j=1,\dots,p$ and $t,t^\prime\in\mathcal{T}$. In this setting, the following assumptions are commonly made to simplify inference: i) \textit{Stationarity}: the functions defining the structural equations, ${f_{i,t}(\cdot,\cdot), i=1,\dots,p}$, $t\in\mathcal{T}$, are time-invariant; ii) \textit{Maximum time lag}: there is a maximum time lag $\tau>0$ above which direct causal effects do not exist, i.e., edges $(v_{j,t^\prime},v_{i,t})$ for $|t-t^\prime|>\tau$ are not permitted; iii) \textit{No backward causality}: edges pointing from the future to the past are prohibited, i.e., there are no edges $(v_{j,t^\prime},v_{i,t})$ with $t^\prime>t$.

\subsection{Partial tail-correlation for regularly varying random vectors}
\label{sec:bg:RegularVarying}

Multivariate regular variation is a common assumption used to describe the tail behavior of a random vector. Essentially, it assumes that the (joint) probability of extreme events converges to a valid limit measure, and asymptotically decays according to a power law at a rate determined by the tail index $\alpha>0$.

\begin{definition}[Multivariate regular variation (\citealp{resnick2007heavy}; Chapter~6)]\label{def:rv}
    A $p$-dimensional random vector $\bm{X}\in \mathbb{R}^p_+$ is regularly varying ($RV$) with tail index $\alpha>0$, denoted by $\bm{X}\in {\rm RV}^p_+(\alpha)$, if there exists a sequence $b_n\to \infty$ such that $n\mathbb{P}\{b_n^{-1}\bm{X} \in\cdot\}\xrightarrow{v} v_{\bm{X}}(\cdot)$ as $n\to\infty$, where $v_{\bm{X}}(\cdot)$ is a Radon measure, which is a Borel measure that is finite on all compact sets and inner regular on $\mathbb{E}^p_+:=[0,\infty]^p\setminus{\{\bm{0}\}}$, and $\xrightarrow{v}$ denotes vague convergence (\citealp{resnick2007heavy}; Chapter~3.3.5).
\end{definition}

The limit measure $v_{\bm{X}}(\cdot)$ satisfies the homogeneity property $v_{\bm{X}}(rB)=r^{-\alpha}v_{\bm{X}}(B)$ for any $r>0$ and any Borel subset $B\subset \mathbb{E}^p_+$. This allows decomposing the limit measure into a radial measure uniquely defined by the tail index $\alpha$ and an angular mass measure $H_{\bm{X}}(\cdot)$, which is defined on the positive part of the unit $(p-1)$-sphere $\mathbb{S}^{p-1}_+:=\{\bm{x}\in \mathbb{R}_+^p: \|\bm{x}\|_2=1\}$, where $\|\cdot\|_2$ denotes the $l^2$ norm. Specifically,
$$
{v_{\bm{X}}(\{\bm{x}\in \mathbb{E}^p_+: \|\bm{x}\|_2\geq r, \bm{x}/\|\bm{x}\|_2\in B_H \})}=r^{-\alpha}H_{\bm{X}}(B_H),
$$ 
for $r>0$ and any Borel subset $B_H \subset\mathbb{S}^{p-1}_+ $. The angular mass measure can be normalized to give a valid probability measure, which is denoted by $N_{\bm{X}}(\cdot):=H_{\bm{X}}(\cdot)/m$ where ${m=\int_{\mathbb{S}^{p-1}_+ }{\rm d}H_{\bm{X}}(\bm{w})}$ is the total mass of $H_{\bm{X}}(\cdot)$.

\cite{kim2022hypothesis} consider (without loss of generality) the case $\alpha=2$ and construct an inner product space  $V^q$ for random variables $X\in {\rm RV}^1_+(2)$ using the transformed-linear operations $\oplus$ and $\circ$, defined as $X_1\oplus X_2=t\{t^{-1}(X_1)+t^{-1}(X_2)\}$ and $a\circ X_1=t\{at^{-1}(X_1)\}$, for $X_1,X_2\in {\rm RV}^1_+(2)$, $a\in \mathbb{R}$, and transformation function ${t(\cdot)=\log\{1+\exp(\cdot)\}}$. When endowed with this transformed-linear algebra, the inner product space preserves regular variation. The space $V^q$ contains elements that can be linearly spanned through the operations $\oplus$ and $\circ$ by a $q$-dimensional random vector $\bm{Z}:=(Z_1,\dots,Z_q)^\top$ with independent components $Z_i\in {\rm RV}^1_+(2)$ satisfying $\lim_{n\rightarrow \infty}n\mathbb{P}\{Z_i<\exp(-\kappa b_n)\}=0$ for $\kappa>0,i=1,\dots,q$. That is,
\[
V^q=\left\{X\in {\rm RV}^1_+(2):X=\bm{a}^\top\circ \bm{Z}=(a_{1}\circ Z_1)\oplus \dots \oplus (a_{q}\circ Z_q), \bm{a}\in \mathbb{R}^q\right\},
\]
where $\bm{a}:=(a_1,\dots,a_q)^\top\in\mathbb{R}^q$ and the inner product is defined as $\langle X, Y \rangle := \bm{a}_x^\top \bm{a}_y$ for $X = \bm{a}_x^\top \circ \bm{Z}$ and $Y = \bm{a}_y^\top \circ \bm{Z}$. In this study, we consider a subspace $V^q_+$ that requires non-negative coefficients $\bm{a}\in \mathbb{R}^q_+$. The subspace $V^q_+$ of regularly varying random vectors is not restrictive as the corresponding angular measure $H_{\mathbf{X}}(\cdot)$ for any $\bm{X}\in {\rm RV}^p_+(2)$ can be approximated using a sequence of transformed-linear combinations of independent regularly varying random variables \cite[Proposition 4; ][]{cooley2019decompositions} with increasing $q$. Furthermore, the inner product in $V^q_+$ can be interpreted as a measure of extremal dependence that is analogous to the covariance between Gaussian random variables. Specifically, the tail pairwise dependence matrix (TPDM;~\citealp{cooley2019decompositions})  $\bm{\Sigma}_{\bm{X}} := (\sigma_{ij})_{i,j=1}^p \in \mathbb{R}_+^{p \times p}$ stores the pairwise inner product between elements of $\bm{X}:=(X_1,\dots,X_{p})^\top\in {\rm RV}^{p}_+(2)$, where $X_{i}\in V^q_+$ for $i=1,\dots,p$. We use $\sigma_{ij}:=\langle X_i,X_j\rangle$ to denote the $(i,j)$-th element of $\bm{\Sigma}_{\bm{X}}$, which admits the conditional limit form
\begin{equation*}%\label{eq:estimationSigma}
    \sigma_{ij}=m\int_{\mathbb{S}^{1}_+} \omega_i\omega_j {\rm d}N_{\bm{X}}(\bm{\omega})=m\lim_{r\to \infty} \mathbb{E}\left[\frac{X_iX_j}{R^2}\mid R>r\right]=m \lim_{r\to \infty} \mathbb{E}[W_iW_j\mid R>r],
\end{equation*}
%where 
with radius $R=\|\bm{X}\|_2=(\sum_{i=1}^p X_i^2)^{1/2}$ and angles ${\bm{W}=(W_1,\dots,W_p)^\top={\bm{X}/ R}}$. This quantity can be estimated empirically using independent samples $\{\bm{x}_l\}_{l=1}^{n}$ of $\bm{X}$ as
\begin{equation}
    \label{eq:estimationSigma11}
    \hat{\sigma}_{ij}=\hat{m}N^{-1}\sum_{l=1}^{n}\frac{x_{li}x_{lj}}{r^2_l}\mathds{1}(r_l>r_0),
\end{equation}
where $r_l=\|\bm{x}_l\|_2$, $r_0>0$ is a suitably-chosen high threshold which can be set to the $q$-quantile (with $q$ close to $1$) of $\{\|\bm{x}_l\|_2\}_{l=1}^{n}$, and $N=\sum_{l=1}^n \mathds{1}(r_l>r_0)$ is the number of threshold exceedances, with $\mathds{1}(\cdot)$ the indicator function. Moreover, $\hat{m}$ is the estimated total mass of $m=H_{\bm{X}}(\mathbb{S}^{p-1}_+)$ and may be calculated as 
\begin{equation}
    \label{eq:estimationm}
    \hat{m}=r_0^{2}n^{-1}\sum_{l=1}^n \mathds{1}(r_l>r_0).
\end{equation}
%~\citep[][Sec.~7.1]{cooley2019decompositions}. 
When $\bm{X}$ is marginally standardized in the sense that $\sigma_{ii}=1$ for $i=1\dots,p$, the total angular mass is $m=p$. Such a standardization can be performed either through a tail index-based transformation when the original data are marginally regularly varying, or via an empirical rank-based transformation, which standardizes the margins to the Pareto distribution with unit scale and shape 2; see e.g., \cite{jiang2024efficient}. The marginal tail index can be estimated using, for example, the~\cite{hill1975simple} estimator.

Analogously to the partial correlation in Gaussian random vectors, \cite{gong2024partial} and \cite{kim2022hypothesis} further introduced the partial tail-correlation coefficient (PTCC) as follows. For any triplet $(X_i,X_j,\bm{\mathcal{S}})$, where ${\bm{\mathcal{S}}\subset\{X_1,\dots,X_p\}\setminus \{X_i,X_j\}}$, the PTCC $\gamma_{ij| \bm{\mathcal{S}}}$ measures the extremal dependence between $X_i$ and $X_j$ after adjusting for the effects of variables in $\bm{\mathcal{S}}$ via transformed-linear optimal prediction~\citep{lee2021transformed}. Let $\bm{X}^\prime:=(X_i,X_j,\bm{\mathcal{S}}^\top)^\top$. The TPDM of $\bm{X}^\prime$ can be written in block matrix form as 
$$
\bm{\Sigma}_{\bm{X}^\prime} =
\left[\begin{smallmatrix}
\bm{\Sigma}_{ij,ij} & \bm{\Sigma}_{ij,\bm{\mathcal{S}}} \\
\bm{\Sigma}_{\bm{\mathcal{S}},ij} & \bm{\Sigma}_{\bm{\mathcal{S}},\bm{\mathcal{S}}}
\end{smallmatrix}\right],
$$ where $\bm{\Sigma}_{ij,ij} \in \mathbb{R}_+^{2\times 2}$ stores the pairwise inner product for components in $(X_i,X_j)^\top$ and a similar definition applies for other block matrices in $\bm{\Sigma}_{\bm{X}^\prime}$. Then, the PTCC $\gamma_{ij| \bm{\mathcal{S}}}$ is defined as the normalized $(1,2)$-th entry of the partial tail-covariance matrix $\bm{\Sigma}_{ij| \bm{\mathcal{S}}}$, given by:
\begin{equation}
    \label{eq:PTCC}
        \bm{\Sigma}_{ij| \bm{\mathcal{S}}}=\bm{\Sigma}_{ij,ij} - \bm{\Sigma}_{ij,\bm{\mathcal{S}}}\bm{\Sigma}_{\bm{\mathcal{S}},\bm{\mathcal{S}}}^{-1}\bm{\Sigma}_{\bm{\mathcal{S}},ij}~\textrm{and}\quad \gamma_{ij| \bm{\mathcal{S}}}= \frac{(\bm{\Sigma}_{ij| \bm{\mathcal{S}}})_{12}}{\sqrt{(\bm{\Sigma}_{ij| \bm{\mathcal{S}}})_{11}(\bm{\Sigma}_{ij| \bm{\mathcal{S}}})_{22}}}\in [-1,1].
\end{equation}
%[NOTE]: equation (9) and (10) in kim2022hypothesis is the same as Definition 2
If $\gamma_{ij| \bm{\mathcal{S}}} = 0$, following~\cite{gong2024partial}, we say that $X_i$ and $X_j$ are partially tail uncorrelated given $\bm{\mathcal{S}}$, and we denote it as $X_i \perp_{\rm TC} X_j\mid \bm{\mathcal{S}}$. From Equation~\eqref{eq:PTCC}, this is equivalent to the partial tail-covariance, $\sigma_{ij| \bm{\mathcal{S}}} := (\bm{\Sigma}_{ij| \bm{\mathcal{S}}})_{12},$ being zero. The estimator of the partial tail-covariance, $\hat{\sigma}_{ij| \bm{\mathcal{S}}}$, can be calculated similarly to \eqref{eq:estimationSigma11}, but from the residual components $\bm{\epsilon}:=(\epsilon_{i},\epsilon_{j})^\top$ obtained from $\bm{X}_{ij}:=(X_{i},X_{j})^\top$ after adjusting for the effect of $\bm{\mathcal{S}}$ using the estimated TPDM $\hat{\bm{\Sigma}}_{\bm{X}^\prime}$. Specifically,
\begin{equation}
    \label{eq:estimationCondSigma12}
    \bm{\epsilon}=t^{-1}(\bm{X}_{ij})-(\hat{\bm{\Sigma}}_{ij,\bm{\mathcal{S}}}\hat{\bm{\Sigma}}_{\bm{\mathcal{S}},\bm{\mathcal{S}}}^{-1})t^{-1}(\bm{\mathcal{S}}),
\end{equation}
where the $\hat{\bm{\Sigma}}_{ij,\bm{\mathcal{S}}}$ and $\hat{\bm{\Sigma}}_{\bm{\mathcal{S}},\bm{\mathcal{S}}}$ are sub-matrices of the estimated TPDM $\hat{\bm{\Sigma}}_{\bm{X}^\prime}$, with each element estimated using~\eqref{eq:estimationSigma11}. Then, $\hat{\sigma}_{ij| \bm{\mathcal{S}}}$ is estimated by
\begin{equation}
    \label{eq:estimationCondSigma11}
        \hat{\sigma}_{ij| \bm{\mathcal{S}}}=\hat{m}^\prime \frac{1}{N^\prime}\sum_{l=1}^{n}\frac{\epsilon_{li}\epsilon_{lj}}{(r^\prime_l)^2}\mathds{1}(r^\prime_l>r^\prime_0)
\end{equation} from the residuals $\{(\epsilon_{li},\epsilon_{lj})\}_{l=1}^n$, where $r^\prime_l$ is now the $l^2$-norm of residuals $\bm{\epsilon}_l$, and $r^\prime_0>0$ is a suitably-chosen high threshold selected as the $q$-quantile ($q$ close to $1$) of $\{\|\bm{\epsilon}_l\|_2\}_{l=1}^n$. The number of threshold exceedances is given by $N^\prime=\sum_{l=1}^n \mathds{1}(r^\prime_l>r^\prime_0)$. Moreover, since the residuals are not standardized, $\hat{m}^\prime$ needs to be estimated using Equation~\eqref{eq:estimationm}. 

The estimator $\hat{\sigma}_{ij \mid \bm{\mathcal{S}}}$ is asymptotically Gaussian under Condition~5.6 of \citet{kim2022hypothesis}, first introduced in \citet[][Theorem 1]{larsson2012extremal} and reproduced in Equation~\eqref{eq:condition5.6}:
\begin{equation}\label{eq:condition5.6}
    \sqrt{k} \left\{ \frac{n}{k} P \left[ \left( \frac{R}{b(n/k)}, \bm{W} \right) \in \cdot \right] 
    - \nu_\alpha \times \frac{1}{m} H_{\bm{X}}(\cdot) \right\} \xrightarrow{v} 0,
    \quad \text{as } n \to \infty ,
\end{equation}

where $\nu_\alpha$ denotes the radial measure defined by $\nu_\alpha(x,\infty] = x^{-\alpha}$ for $x>0$.
Intuitively, this condition requires that the dependence between the radial component $R$ and the angular component $\bm{W}$ becomes asymptotically negligible as the magnitude of $R$ increases. In particular, if $R$ and $\bm{W}$ are independent, the condition is automatically satisfied.

Therefore, under condition~\eqref{eq:condition5.6}, standard $t$-tests can be performed for testing the hypothesis that $\sigma_{ij| \bm{\mathcal{S}}}=0$ and $X_i \perp_{\rm TC} X_j\mid \bm{\mathcal{S}}$.

\section{Methodology}
\label{sec:extremal_causality}

In Section~\ref{sec:extremal_causality:modeling}, we construct the XSCM to model causal relationships among extremal events using the transformed-linear operations described in Section~\ref{sec:bg:RegularVarying}, and in Section~\ref{sec:extremal_causality:MarkovFaithfulness} we then discuss the connection between partial tail uncorrelation in the random vector and d-separation in the corresponding graph. The extension of the XSCM to a time series setting, and its inference using separation-based methods are presented in Sections~\ref{sec:extremal_causality:timeSeries} and~\ref{sec:extremal_causality:inference}, respectively. We further discuss how our framework can be applied to the undirected graph setting in Section~\ref{sec:extremal_causality:markovNetwork}.

\subsection{Structural causal model for extremes}
\label{sec:extremal_causality:modeling}

We define the XSCM for a multivariate regularly varying random vector ${\bm{X}={(X_1,\dots,X_p)}^\top}$ with tail index $\alpha=2$. This assumption facilitates the modeling of tail-dependence and it is without loss of generality as far as marginal transformations are concerned; see the discussion on marginal transformations in Section~\ref{sec:bg:RegularVarying}.

\begin{definition}[Extremal structural causal model (XSCM)]
    \label{def:model:transformedLinearCausality}
    An XSCM is defined as ${\mathcal{S}:=(\bm{Z},\mathcal{F})}$, where the source variables in $\bm{Z}:=(Z_1,\dots,Z_p)^\top$ are independent and identically distributed with $Z_i\in {\rm RV}^1_+(2),i=1,\dots,p$, and the collection $\mathcal{F}$ consists of $p$ structural equations in transformed-linear form:
    \begin{equation}
        \label{eq:transformedLinear}
        X_i=f_{X_i}({\rm Pa}(X_i),Z_i)=(\alpha_i\circ Z_i)\oplus \bigoplus_{X_j\in {\rm Pa}(X_i)}(\beta_{j\rightarrow i}\circ X_j),
    \end{equation}
    for $i=1,\dots,p$, where $\beta_{j\rightarrow i}>0,~\alpha_i>0,i,j=1,\dots,p$.
\end{definition}

Here, the transformed-linear operators \(\oplus\) and \(\circ\) are the same as those defined in Section~\ref{sec:bg:RegularVarying}. The XSCM can be expressed in the matrix form as $\bm{X}=\bm{A}\circ \bm{Z} \oplus \bm{B}\circ \bm{X}$ where the \textit{path coefficient matrix} $\bm{B}=(B_{ij})_{i,j=1}^p\in \mathbb{R}^{p\times p}_+$ consists of non-negative entries, with $B_{ij}=\beta_{j\rightarrow i}$ if $(v_j,v_i)\in \mathcal{E}$, and $B_{ij}=0$ otherwise. The \textit{scale matrix} $\bm{A}\in\mathbb{R}^{p\times p}_+$ is defined as $\bm{A}:={\rm diag}(\alpha_1,\dots,\alpha_p)$. The random vector $\bm{X}$ is regularly varying, i.e., $\bm{X}\in {\rm RV}^p_+(2)$, with each component satisfying $X_{i} \in V^q_+$ for $i=1,\dots,p$, and the TPDM of $\bm{X}$ can be obtained using Proposition~\ref{prop:tpdmForOnetails}, proved in Appendix~\ref{sec:Proof}.

\begin{proposition}[TPDM $\bm{\Sigma}_{\bm{X}}$ of the XSCM]\label{prop:tpdmForOnetails} 
    Let $\mathcal{S}:=(\bm{Z},\mathcal{F})$ be an XSCM with path coefficient matrix $\bm{B}$, scale matrix $\bm{A}$, and resulting random vector ${\bm{X}\in {\rm RV}^p_+(2)}$. Then, a direct representation for $\bm{X}$ is $\bm{X}=(\bm{I}-\bm{B})^{-1}\bm{A}\circ \bm{Z}$, where ${(\bm{I}-\bm{B})^{-1}\in\mathbb{R}^{p\times p}_+}$ is a (non-singular) non-negative matrix, and the TPDM of $\bm{X}$ is given by $\bm{\Sigma}_{\bm{X}}={(\bm{I}-\bm{B})}^{-1}\bm{A}^{2}{[{(\bm{I}-\bm{B})}^{-1}]}^\top,$ where $\bm{I}$ is the $p\times p$ identity matrix.
\end{proposition}

Proposition~\ref{prop:tpdmForOnetails} shows that $\bm{X}$ admits the representation ${\bm{X} = (\bm{I} - \bm{B})^{-1} \bm{A} \circ \bm{Z}}$, where the matrices ${(\bm{I}-\bm{B})}^{-1}$ and $\bm{A}$ are non-negative. Thus, each $X_i$ is a transformed-linear combination of the elements of $\bm{Z}$ with non-negative coefficients. 

\subsection{Equivalence of zero PTCC and d-separation in the XSCM}
\label{sec:extremal_causality:MarkovFaithfulness}

A key advantage of the XSCM over existing causal models for heavy-tailed random variables is that it satisfies extremal analogues of the causal Markov and causal faithfulness conditions; we refer to these as the \textit{tail causal Markov} and \textit{tail causal faithfulness} conditions. These two properties establish an equivalence between partial tail uncorrelation in the random vector $\bm{X}$ and d-separation in its corresponding graph $\mathcal{G}$, which facilitates the estimation of the underlying graph structure.

\begin{proposition}[Tail causal Markov and tail causal faithfulness]\label{prop:Mar}    
    Given an XSCM for $\bm{X}=(X_1,\dots,X_p)^\top$ with underlying graph $\mathcal{G}=(\mathcal{V},\mathcal{E})$, for any two vertices $v_i,v_j \in \mathcal{V}$ and a separation set $\bm{\mathcal{S}}_v\subset \mathcal{V}\setminus \{v_i,v_j\}$, the following holds: $v_i \perp_{\mathcal{G}}  v_j \mid \bm{\mathcal{S}}_v$ if and only if $\gamma_{ij|\bm{\mathcal{S}}} = 0$, where $\bm{\mathcal{S}}:=\{X_k\}_{v_k\in\bm{\mathcal{S}}_v}$, i.e., $X_i\perp_{\rm TC} X_j \mid \bm{\mathcal{S}}$.
\end{proposition}

The proof of Proposition~\ref{prop:Mar} is provided in Appendix~\ref{sec:Proof}. The main idea is to demonstrate that for any XSCM $\mathcal{S}$ defined for $\bm{X}=(X_1,\dots,X_p)^\top\in {\rm RV}^p_+(2)$, there exists a corresponding Gaussian linear SCM $\mathcal{S}^\prime$ for a Gaussian $\bm{X}^\prime=(X_1^\prime,\dots,X_p^\prime)^\top$. Then, as $\mathcal{S}$ and $\mathcal{S}^\prime$ share the same path coefficient matrix $\bm{B}$ and scale matrix $\bm{A}$,  they induce the same graph $\mathcal{G}=(\mathcal{V},\mathcal{E})$. 

For any $v_i,v_j\in \mathcal{V}$ and a separation set $\bm{\mathcal{S}}_v\subset \mathcal{V}\setminus \{v_i,v_j\}$, the partial tail-covariance, $\sigma_{ij|\bm{\mathcal{S}}}$, of $\bm{X}$ is equal to the conditional covariance ${\rm Cov}[X_i,X_j\mid\bm{\mathcal{S}}^\prime]$ of $\bm{X}^\prime$, where $\bm{\mathcal{S}}:=\{X_k\}_{v_k\in\bm{\mathcal{S}}_v}$ and $\bm{\mathcal{S}}^\prime:=\{X_k^\prime\}_{v_k\in\bm{\mathcal{S}}_v}$. Note that we are not seeking equivalence between d-separation and conditional independence, but instead equivalence between d-separation, $v_i\perp_\mathcal{G} v_j\mid\bm{\mathcal{S}}_v$, and partial tail uncorrelation, $X_i\perp_{\rm TC} X_j\mid \bm{\mathcal{S}}$.

\subsection{XSCM for time series}\label{sec:extremal_causality:timeSeries}

The proposed XSCM framework can also  be readily adapted to time series data. Here we assume that the XSCM is stationary, has a finite maximum causal lag $\tau>0$, and no backward causality, as discussed in Section~\ref{sec:causal:causal_graph_model}.
Similarly to Definition~\ref{def:model:transformedLinearCausality}, the XSCM for a time series $\{\bm{X}_t: = (X_{1,t}, \dots, X_{p,t})^\top\}_{t\in \mathcal{T}}$, where $\mathcal{T}$ denotes the time index set, defines each variable using the transformed-linear structural equation
\[
X_{i,t} = (\alpha_i \circ Z_{i,t}) \oplus \bigoplus_{X_{j,t-\delta} \in \mathrm{Pa}(X_{i,t}); \delta=0,\dots,\tau} \left( \beta^\delta_{j \rightarrow i} \circ X_{j,t-\delta} \right), \quad i=1,\dots,p,\quad t\in \mathcal{T},
\]
where $\alpha_i>0, \beta^\delta_{j \rightarrow i} > 0$. The model can be compactly written in matrix form as:
\begin{equation}\label{eq:XSCM:timeSeries}
    \bm{X}_t = \bm{B}_{(0)} \circ \bm{X}_{t} \oplus \bm{B}_{(1)} \circ \bm{X}_{t-1} \oplus \dots \oplus \bm{B}_{(\tau)} \circ \bm{X}_{t-\tau} \oplus \bm{A} \circ \bm{Z}_t,
\end{equation}
where path coefficient matrices $\bm B_{(\delta)}=(B_{(\delta);ij})_{i,j=1}^p\in\mathbb{R}_+^{p\times p}$ captures contemporaneous ($\delta=0$) and lagged (${\delta=1,\ldots,\tau}$) causal effects. The entry $B_{(\delta);ij}$ denotes the strength of the lag-$\delta$ causal influence from $X_{j,t-\delta}$ to $X_{i,t}$. The direct form of the XSCM is given by:
\begin{equation}\label{eq:XSCM:timeSeries:direct}
    \bm{X}_t = (\bm{I} - \bm{B}_{(0)})^{-1} \left( \bm{B}_{(1)} \circ \bm{X}_{t-1} \oplus \dots \oplus \bm{B}_{(\tau)} \circ \bm{X}_{t-\tau} \oplus \bm{A} \circ \bm{Z}_t \right).
\end{equation}

Again, the matrix $(\bm{I} - \bm{B}_{(0)})^{-1}\in \mathbb{R}^{p\times p}_+$ is non-singular and non-negative and guarantees that each $X_{i,t}\in V_+^q$ for $i=1,\dots,p,t\in \mathcal{T}$. When the maximum causal time lag, $\tau$, is set to zero,~\eqref{eq:XSCM:timeSeries:direct} simplifies to the XSCM in Definition~\ref{def:model:transformedLinearCausality}.

\subsection{Separation-based graphical learning}
\label{sec:extremal_causality:inference}

Proposition~\ref{prop:Mar} shows that, for XSCMs, testing whether the PTCC equals zero is equivalent to testing for d-separation in the underlying DAG. Thus, we can exploit many of the popular constraint-based methods for learning the underlying graph of the XSCM. We choose the PCMCI$^+$ algorithm~\citep{runge2020discovering}, which improves upon the PC algorithm (named after its inventors,~\citealp{spirtes2000causation}), and can be easily applied to time series data. Note that other constraint-based algorithms, such as those by~\cite{spirtes2001anytime} and~\cite{colombo2012learning}, can also be combined with our proposed separation-based methodology.

The graph learning procedure from the PCMCI$^+$ algorithm consists of two phases: i) \textit{skeleton estimation} and ii) \textit{orientation}. Skeleton estimation identifies the undirected structure of the underlying graph $\mathcal{G}:=(\mathcal{V},\mathcal{E})$ defining an XSCM for the vector ${\bm{X}=(X_1,\dots,X_p)^\top}$ by iteratively removing edges from a fully-connected undirected graph. For any two vertices $v_i,v_j$ in $\mathcal{V}$, the edges $(v_i,v_j)$ and $(v_j,v_i)$ are removed from $\mathcal{E}$ if there exists a set $\bm{\mathcal{S}}_v \subset \mathcal{V}\setminus \{v_i,v_j\}$ such that $v_i \perp_\mathcal{G} v_j \mid\bm{\mathcal{S}}_v$, i.e., $X_i\perp_{\rm TC}X_j \mid \bm{\mathcal{S}}$ with $\bm{\mathcal{S}}=\{X_k\}_{v_k\in \bm{\mathcal{S}}_v}$. The \textit{orientation phase} determines the directions of all remaining edges after the skeleton has been estimated by using the facts that colliders cannot appear in the separation set and that DAGs are acyclic.

Partial tail uncorrelation can be tested via a standard $t$-test as discussed in Section~\ref{sec:bg:RegularVarying} and \cite{kim2022hypothesis}. We implement the separation test by using the Python package \texttt{Tigramite}~\citep{runge2020discovering}. The consistency, up to Markov equivalence, of underlying graph estimation by the PCMCI$^+$ algorithm has been established in Theorems~1--3 of \citet{runge2020discovering} under the four standard assumptions of causal discovery: i) causal sufficiency (no unobserved common causal parents), ii) the causal Markov condition, iii) causal faithfulness, and iv) the availability of a consistent conditional independence test. Proposition~\ref{prop:Mar} shows that XSCMs satisfy extremal analogues of the causal Markov and causal faithfulness conditions. Moreover, the PTCC-based separation test is consistent under Condition~\eqref{eq:condition5.6}.

Therefore, under causal sufficiency assumption and Condition~\eqref{eq:condition5.6}, the PCMCI$^+$ algorithm equipped with PTCC-based separation tests consistently recovers the underlying DAG structure generated by an XSCM, up to Markov equivalence, as the sample size increases.

Hyperparameters for our graph learning procedure include the threshold level $q\in (0,1)$ in estimating the partial tail-covariance $\sigma_{ij|\bm{\mathcal{S}}}$ in \eqref{eq:estimationCondSigma11}, and the significance level $\alpha\in (0,1)$ for testing the hypothesis ${\mathcal{H}_0: X_i \perp_{\rm TC} X_j \mid \bm{\mathcal{S}}}$. In addition, for time series, the parameter $\tau$ determining the maximum allowed causal time lag also needs to be specified. To approximately satisfy the limit in Condition~\eqref{eq:condition5.6}, the threshold level $q$ should be chosen sufficiently close to~1, while ensuring that the number of threshold exceedances $N^\prime$ remains large enough for reliable estimation of the partial tail-covariance. This is the classical trade-off dilemma between bias and variance in extreme value theory. In practice, the values of $\alpha$ and $q$ can be tuned to achieve a desirable level of sparsity in the estimated graph.

\subsection{Connection with undirected graph learning for extremes}\label{sec:extremal_causality:markovNetwork}

Undirected graphs, in which edges encode extremal associations between variables, are also key for the modeling of multivariate extremes. We here show that, for the undirected graphical frameworks introduced by~\cite{gong2024partial} and~\cite{kim2022hypothesis}, the underlying graph structure can also be learned using the proposed separation-based methods described in Section~\ref{sec:extremal_causality:inference}. Specifically, the intermediate phase consisting of estimating the skeleton of the graph provides a valid estimate of the underlying undirected graph when $\bm{X}$ is assumed to follow an extremal Markov network. We now define undirected graphical models for regularly varying random vectors $\bm{X} \in {\rm RV}^p_+(2)$, with $X_i \in V^q_+$ for $i=1,\dots,p$.

\begin{definition}[Extremal Markov network (\citealp{gong2024partial}, Sec.~4.1;~\citealp{kim2022hypothesis})]\label{def:markovNetwork}
Let $\bm{X} = (X_1, \dots, X_p)^\top$ be a random vector with $X_i\in V_+^q$ for each $i=1,\dots,p$. Let $\bm{\Sigma}_{\bm{X}}$ denote the TPDM of $\bm{X}$, and let $\bm{\Sigma}^{-1}_{\bm{X}}$ be its inverse matrix. The random vector $\bm{X}$ is called an \textit{extremal Markov network} with respect to the induced undirected graph $\check{\mathcal{G}} = (\mathcal{V}, \mathcal{E})$, with vertex set $\mathcal{V} = \{v_1, \dots, v_p\}$ and edge set  
\[
\mathcal{E} := \left\{ (v_i, v_j) \in \mathcal{V} \times \mathcal{V} : (\bm{\Sigma}^{-1}_{\bm{X}})_{ij} \neq 0, ~ i \neq j \right\}.
\]
\end{definition}

Extremal Markov networks satisfy the pairwise Markov property~\citep[][Proposition~1]{gong2024partial}, which states that $(v_i,v_j)$ is not an edge in $\check{\mathcal{G}}$ if and only if $\gamma_{ij|\bm{\mathcal{S}}} = 0$, where $\bm{\mathcal{S}} = \{X_1, \dots, X_p\} \setminus \{X_i, X_j\}$. To generalize separation-based estimation to undirected graphs, we need to establish the global Markov property. To this end, we first recall the definition of separation in undirected graphs.

\begin{definition}[Separation in an undirected graph]\label{def:separationMar}
    Let $\check{\mathcal{G}} = (\mathcal{V}, \mathcal{E})$ be an undirected graph. Two vertices $v_i, v_j \in \mathcal{V}$ are said to be separated given a subset $\bm{\mathcal{S}}_v \subset \mathcal{V} \setminus \{v_i, v_j\}$ if every path between $v_i$ and $v_j$ is blocked by $\bm{\mathcal{S}}_v$, in the sense that every such path passes through at least one vertex in $\bm{\mathcal{S}}_v$.
\end{definition}

The concept of separation in undirected graphs is simpler than d-separation in DAGs, as it depends solely on path obstruction and does not involve colliders. We extend the pairwise Markov property by showing that extremal Markov networks also satisfy the \textit{global Markov property}.

\begin{proposition}[Global Markov property for extremal Markov networks]\label{prop:extrGlobalMarkov}
Assume that $\bm{X} = (X_1, \dots, X_p)^\top$ is an extremal Markov network with $X_{i} \in V_+^q$ for ${i=1,\dots,p}$, and let $\check{\mathcal{G}} = (\mathcal{V}, \mathcal{E})$ be its induced undirected graph. For any two vertices $v_i, v_j \in \mathcal{V}$ and any separation set ${\bm{\mathcal{S}}_v \subset \mathcal{V} \setminus \{v_i, v_j\}}$,  $v_i$ and $v_j$ are separated given $\bm{\mathcal{S}}_v$ if and only if $\gamma_{ij|\bm{\mathcal{S}}} = 0$, i.e, $X_i\perp_{\rm TC} X_j \mid \bm{\mathcal{S}}$, where $\bm{\mathcal{S}} = \{X_k : v_k \in \bm{\mathcal{S}}_v\}$.
\end{proposition}

The proof is provided in Appendix~\ref{sec:Proof}.
Proposition~\ref{prop:extrGlobalMarkov} establishes equivalence between separation in the undirected graph of an extremal Markov network and zero PTCC values. The underlying graph of an extremal Markov network can thus be inferred by removing edges $(v_i,v_j)$ if there exists a separation set $\bm{\mathcal{S}}_v \subset \mathcal{V}\setminus \{v_i,v_j\}$, which is the output of the skeleton estimation phase of the PCMCI$^+$ algorithm. Therefore, graph learning for extremal Markov networks can be achieved by running the first phase (skeleton estimation) only of PCMCI$^+$, which is based on performing PTCC tests for all possible separation sets.

\section{Simulation studies}\label{sec:simulation}

In this section, we assess the efficiency of our separation-based method for recovering the underlying graph structure from simulated regularly-varying (i.e., heavy-tailed and tail-dependent) data. We demonstrate broad applicability of our framework across three distinct scenarios: i) estimating the DAG of an XSCM model using cross-sectional (static) data~(Section~\ref{sec:sim:XSCM:non-time});
ii) estimating the undirected graph of an extremal Markov network using cross-sectional data~(Section~\ref{sec:sim:extremalNetwork}); and iii) estimating the DAG of an XSCM model for time series data~(Section~\ref{sec:sim:XSCM:time-series}). Section~\ref{sec:sim:overview} provides an overview of the simulation framework. Besides the scenarios where data are generated from the exact XSCM model class introduced in this paper, we also investigate the performance of the proposed method when the data are generated from other models. Specifically, we consider two classical models for extremes: the max-linear model for DAG estimation and the multivariate H\"usler--Reiss model for undirected graph estimation. The results suggest that our separation-based approach remains stable under model misspecification in these two scenarios. The corresponding results are reported in Appendix~\ref{sec:sim:XSCM:misspecific} and~\ref{sec:sim:extremalNetwork:misspecification}, respectively.

\subsection{Overview}\label{sec:sim:overview}

In all three scenarios, we follow a simulation framework comprising three steps: i) initialization of the graphical model with a controlled sparsity level $\phi\in (0,1)$; ii) data simulation and visualization given an initialized graphical model; and iii) estimation of the underlying graph structure based on the simulated data and evaluation of estimation accuracy. The initialization details for each model are presented in their respective sections. Here, we describe the common settings for data simulation, graph estimation, and accuracy evaluation.

For simulation, given an initialized XSCM, we fix the scale matrix for the source variables $\bm{Z}$ (or $\bm{Z}_t$ for time series) as $\bm{A} := \mathrm{diag}(1,\dots,1)$. The source variables consists of i.i.d.\ components following a Pareto distribution with unit scale and shape parameter 2, simulated as $F^{-1}_{Z}(U)$ with $F_Z(z)=1-z^{-2}, z>1$, and $U\sim \mathrm{Unif}(0,1)$.

To assess the accuracy of edge recovery in the estimated graph structure, we randomly initialize $M = 50$ graphical models and generate $n$ i.i.d.\ samples from each model. We then apply our method to each of the $M$ models to estimate the underlying graph structure. The accuracy of edge recovery is assessed using the normalized edit distance (NED), defined as 
\begin{equation}\label{eq:edit}
    \begin{aligned}
        \mbox{NED}(\mathcal{E},\hat{\mathcal{E}})=\frac{\sum_{e\in \mathcal{E}}\mathds{1}(e\notin \hat{\mathcal{E}})+\sum_{e\in \hat{\mathcal{E}}}\mathds{1}(e\notin \mathcal{E})}{\|\mathcal{E}\|+\|\hat{\mathcal{E}}\|},
    \end{aligned}
\end{equation}
where $\|\mathcal{E}\|$ and $\|\hat{\mathcal{E}}\|$ denote the number of edges in the true and estimated graphs, respectively. We also consider a variant that ignores edge directions—the undirected normalized edit distance (UNED)—defined as
\begin{equation}\label{eq:edit_undirected}
    \begin{aligned}
        \mbox{UNED}(\mathcal{E},\hat{\mathcal{E}})=\mbox{NED}(\mathcal{E}^\prime,\hat{\mathcal{E}}^\prime),
    \end{aligned}
\end{equation}
where $\mathcal{E}^\prime = \{(v_i, v_j) : (v_i, v_j) \in \mathcal{E} \text{ or } (v_j, v_i) \in \mathcal{E}\}$ (with a similar definition for $\hat{\mathcal{E}}^\prime$). Since some directions are unidentifiable via d-separation, we further consider a variant that excludes such edges. This metric, denoted NED$^*$, is defined as
\begin{equation}\label{eq:edit_undirected_star}
    \begin{aligned}
        \mbox{NED}^*(\mathcal{E},\hat{\mathcal{E}})=\mbox{NED}(\mathcal{E}^*,\hat{\mathcal{E}}^*),
    \end{aligned}
\end{equation}
where $\mathcal{E}^*$ and $\hat{\mathcal{E}}^*$ contain only the edges whose directionality is identifiable via d-separation. All metrics—NED, UNED, and NED$^*$—range between 0 and 1, with smaller values indicating more accurate graph recovery.

For the competing methods, there are many candidates for graphical estimation in extremes, including \citet{wan2023graphical}, \citet{tran2024estimating}, \citet{lederer2023extremes}, \cite{krali2025causal}, \cite{engelke2022structure}, \cite{engelke2025extremes}, \cite{bodik2024granger}, \cite{gong2024partial}, and \citet{engelke2021learning}. Our selection of comparison methods is guided by their relevance to the problem settings considered here and the availability of publicly available code to implement the methods.

For the scenario of estimating the DAG of an XSCM using cross-sectional data, apart from two recent papers by \citet{engelke2025extremes} and \citet{krali2025causal}, no other methods are currently available. Since \citet{krali2025causal} does not provide publicly available code, we do not conduct a direct comparison. Instead, we perform a robustness check of DAG estimation when the data are generated from a max-linear model; see Appendix~\ref{sec:sim:XSCM:misspecific} for more details. The approach of \citet{engelke2025extremes} focuses primarily on a pruning algorithm and is therefore not directly comparable to the method proposed in this paper. For the scenario of estimating the undirected graph of an extremal Markov network using cross-sectional data, methods in \cite{engelke2022structure}, \cite{kim2022hypothesis}, and \cite{gong2024partial} are available for comparison. We choose to compare our methods with that proposed by~\cite{gong2024partial}, as it performs best in recovering the true physical graph topology in the Danube river application among the three methods, with the fewest false and missed edges relative to the physical river topology; see Figure~\ref{fig:manyResultsOnDanube} in Section~\ref{sec:app:river}. For the scenario of estimating the DAG of an XSCM model using time series data, we compare our method with that of~\cite{bodik2024granger}.

We optimize the hyperparameters for the method in \cite{gong2024partial}, as no default values are provided, but use the default settings for the method in~\cite{bodik2024causality}. Throughout the remainder of the paper, the parameters in our method are fixed as follows: $\alpha = 0.005$ and $q = 0.99$. Sensitivity analyses regarding the choices of $\alpha$ and $q$ are provided in Appendix~\ref{sec:sensitivity}.

\subsection{Estimating DAGs for XSCMs with cross-sectional data}\label{sec:sim:XSCM:non-time}

To initialize an XSCM for the vector $\bm{X}:=(X_1,\dots,X_p)^\top$, we randomly generate a path coefficient matrix $\bm{B}$ with a predefined sparsity level $\phi\in(0,1)$. For simplicity, we assume that the components in $\bm{X}$ are topologically ordered, so that the corresponding path coefficient matrix $\bm{B}$ is a lower triangular matrix (see Appendix~\ref{sec:Proof}). Specifically, for each $i < j$, the entry $B_{ij}$ is set to zero with probability $1 - \phi$, and otherwise sampled uniformly on the interval $[0,1]$. The associated DAG, $\mathcal{G} = (\mathcal{V}, \mathcal{E})$, is then constructed with vertex set $\mathcal{V} = \{v_1,\dots,v_p\}$ and directed edge set $\mathcal{E} = \{(v_j, v_i): B_{ij} > 0\}$. The left panel of Figure~\ref{fig:XSCM:illustration} shows an example of a randomly initialized $\bm{B}$, while the right panel illustrates the corresponding DAG.

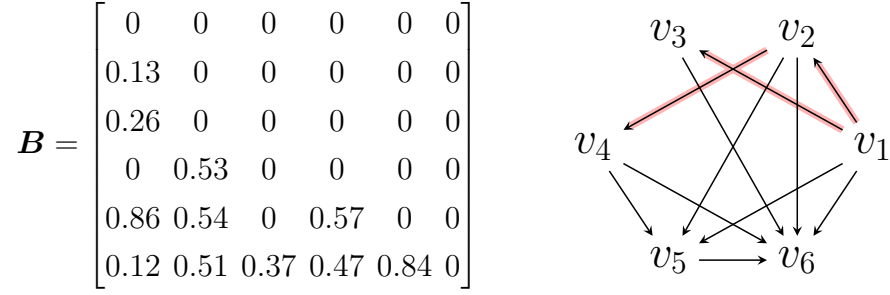
\begin{figure}
    \centering
    \begin{minipage}{0.48\textwidth}
        \vspace{-0.5cm}
        \centering
\[
\renewcommand{\arraystretch}{0.7} 
\bm{B} =
\left[
\begin{array}{@{}c@{\hspace{0.4em}}c@{\hspace{0.4em}}c@{\hspace{0.4em}}c@{\hspace{0.4em}}c@{\hspace{0.4em}}c@{}}
0    & 0    & 0    & 0    & 0    & 0    \\
0.13 & 0    & 0    & 0    & 0    & 0    \\
0.26 & 0    & 0    & 0    & 0    & 0    \\
0    & 0.53 & 0    & 0    & 0    & 0    \\
0.86 & 0.54 & 0    & 0.57 & 0    & 0    \\
0.12 & 0.51 & 0.37 & 0.47 & 0.84 & 0    
\end{array}
\right]
\]  
    \end{minipage}
    \hspace{-1cm}
    \begin{minipage}{0.4\textwidth}
        \centering
        \begin{tikzpicture}
            \node (v1) at (3.7,1.5)  {\Large $v_1$};
            \node (v2) at (2.7,3)  {\Large $v_2$};
            \node (v3) at (1,3)    {\Large $v_3$};
            \node (v4) at (0,1.5)    {\Large $v_4$};
            \node (v5) at (1,0)   {\Large $v_5$};
            \node (v6) at (2.7,0) {\Large $v_6$};   
            \draw[-, line width=1.0mm, draw=red!30, shorten >=3pt] (v2) -- (v4); 
            \draw[-, line width=1.0mm, draw=red!30, shorten >=3pt] (v1) -- (v3); 
            \draw[-, line width=1.0mm, draw=red!30, shorten >=3pt] (v1) -- (v2); 
            \draw[->, line width=0.2mm, >=stealth] (v1) -- (v2); 
            \draw[->, line width=0.2mm, >=stealth] (v1) -- (v3); 
            \draw[->, line width=0.2mm, >=stealth] (v1) -- (v5); 
            \draw[->, line width=0.2mm, >=stealth] (v1) -- (v6); 
            \draw[->, line width=0.2mm, >=stealth] (v2) -- (v4); 
            \draw[->, line width=0.2mm, >=stealth] (v2) -- (v5); 
            \draw[->, line width=0.2mm, >=stealth] (v2) -- (v6); 
            \draw[->, line width=0.2mm, >=stealth] (v3) -- (v6); 
            \draw[->, line width=0.2mm, >=stealth] (v4) -- (v5); 
            \draw[->, line width=0.2mm, >=stealth] (v4) -- (v6); 
            \draw[->, line width=0.2mm, >=stealth] (v5) -- (v6); 
        \end{tikzpicture}
    \end{minipage} 
    \caption{A randomly initialized path coefficient matrix $\bm{B}$ for an XSCM that satisfies topological ordering (left) and its corresponding DAG structure (right). Edges whose directionality cannot be determined by d-separation are highlighted in red.}
\label{fig:XSCM:illustration}
\end{figure}

Given the initialized XSCM, independent samples are generated using the direct form ${\bm{X} = (\bm{I} - \bm{B})^{-1} \circ \bm{Z}}$ where the source variables, $\bm{Z}\in {\rm RV}^p_+$, are simulated as described in Section~\ref{sec:sim:overview}; illustrative scatter plots of samples generated from the XSCM are shown in Figure~\ref{fig:XSCM:scatter} of Appendix~\ref{sec:appendix:simulation:DAGnontime}.

% To evaluate our method, we randomly initialize $M = 50$ XSCMs with varying path coefficient matrix $B_i$ for $i=1,\dots,M$ and generate $n$ i.i.d.\ samples for each case. For each of $i=1,\dots,M$, we apply our method to estimate the underlying DAG and compare the estimated edge set, $\hat{\mathcal{E}}$, with the true edge set, $\mathcal{E}$.

To examine the effect of the number of variables $p$, the connectivity level $\phi$, and the sample size $n$, on graph recovery, we test our method under various $(p, \phi, n)$ combinations. The expected number of edges, $p(p-1)/2 \times \phi$, ranges from 5 to 70. The results are shown in Figure~\ref{fig:DAG_evaluation}, with the left panel showing the method's performance with a small sample size ($n = 5000$) and the right with a large sample size (${n = 50000}$). When ignoring edge directions (UNED), our method performs consistently well, indicating reliable skeleton recovery. When directions are considered (NED and NED$^*$), the error increases. Notably, NED$^*$ improves with larger sample sizes, whereas NED does not. This is consistent with the discussion in Section~\ref{sec:causal:causal_graph_model}, where we highlight that different Markov equivalent DAGs can share the same skeleton and separation sets, leading to ambiguity in direction recovery.
\begin{figure}[t]
    \centering
    \begin{subfigure}{0.475\textwidth}
        \centering
        \includegraphics[width=\textwidth]{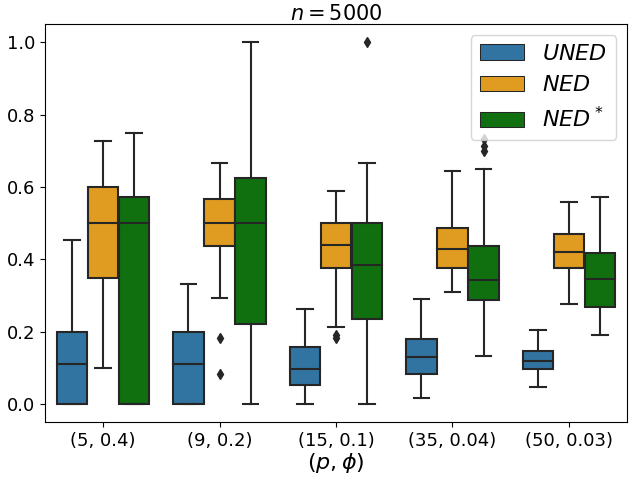}
    \end{subfigure}
    \hspace{0.5cm}
    % 右边的子图
    \begin{subfigure}{0.475\textwidth}
        \centering
        \includegraphics[width=\textwidth]{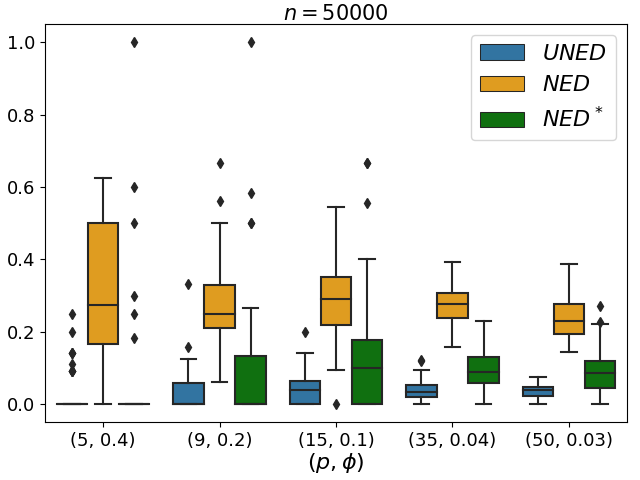} 
    \end{subfigure}
    \caption{Box plots of error measured by NED, UNED, and NED$^*$ for sample size $n=5000$ (left) and $n=50000$ (right)  under varying scale and connectivity parameters $(p, \phi)$. Each setting is repeated for $M = 50$ randomly initialized XSCMs. }
    \label{fig:DAG_evaluation}
\end{figure}

\subsection{Estimating the undirected graph of extremal Markov networks}\label{sec:sim:extremalNetwork}

To initialize an extremal Markov network for $\bm{X} = (X_1, \dots, X_p)^\top$, we first generate a precision matrix $\bm{Q} = \bm{\Sigma}_{\bm{X}}^{-1}$ with varying levels of sparsity $\phi\in(0,1)$, where $\bm{\Sigma}_{\bm{X}}$ denotes the TPDM of $\bm{X}$. To ensure that $\bm{\Sigma}_{\bm{X}}$ is symmetric, positive-definite and non-negative, we require $\bm{Q}=(Q_{ij})_{i,j=1}^p$ to be a symmetric, non-singular $\mathcal{M}$-matrix (see Appendix~\ref{sec:Proof}). 
Specifically, we set each entry ${Q_{ij} = Q_{ji}}$, with $i \neq j$, to zero with probability $\phi$, and otherwise draw it from a uniform distribution over $[-1, 0]$. We then set each diagonal entry to be strictly greater than the maximum absolute value of its corresponding row and column entries. The associated undirected graph $\check{\mathcal{G}} = (\mathcal{V}, \mathcal{E})$ is then defined with vertex set $\mathcal{V} = \{v_1, \dots, v_p\}$ and edge set ${\mathcal{E} = \{(v_i, v_j): Q_{ij} < 0\}}$.

To generate samples from the extremal Markov network $\bm{X} = (X_1, \dots, X_p)^\top$, we use the transformation $\bm{X} = \bm{L} \circ \bm{Z}$, where $\bm{L}$ is the lower triangular matrix from the Cholesky decomposition of $\bm{\Sigma}_{\bm{X}}=\bm{LL}^{\top}$ and  $\bm{Z}\in {\rm RV}^p_+$ is simulated as described in Section~\ref{sec:sim:overview}. An illustration of an extremal Markov network and the corresponding simulated samples are shown in Appendix~\ref{sec:appendix:simulation:emn}. 

%To evaluate the accuracy of our method, we initialize $M = 50$ extremal Markov networks, each with a randomly generated precision matrix $\{Q_i\}_{i=1}^M$, and generate $n$ i.i.d.\ samples for each model. We then estimate the underlying undirected graph from the data and compare the estimated edge set $\hat{\mathcal{E}}$ with the ground truth $\mathcal{E}$ using the normalized edit distance $NED(\mathcal{E}, \hat{\mathcal{E}})$ defined in Equation~\eqref{eq:edit_undirected}. 

Results under varying combinations of scale and connectivity parameters $(p, \phi)$, and two different sample sizes, are summarized in Figure~\ref{fig:extremalmarkov:performance}. The results indicate that our method generally outperforms that proposed by~\cite{gong2024partial}, particularly under high dimensionality and small sample size. The reason why our method generally performs better is that the approach of~\cite{gong2024partial} relies on the pairwise Markov property, which requires conditioning on all other variables as the separation set. In contrast, our method only needs to identify an often much smaller separation set. Conditioning on an unnecessarily large set can introduce noise, making the test less data-efficient and reducing practical performance.

\begin{figure}
    \centering
    \begin{subfigure}{0.475\textwidth}
        \centering
        \includegraphics[width=\textwidth]{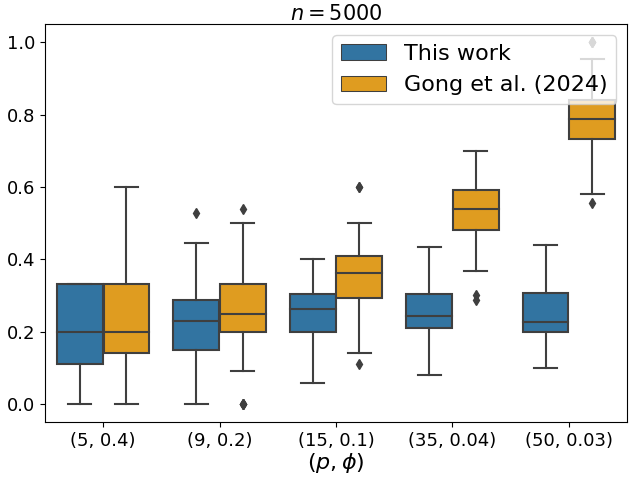}
    \end{subfigure}
    \hspace{0.5cm}
    % 右边的子图
    \begin{subfigure}{0.475\textwidth}
        \includegraphics[width=\textwidth]{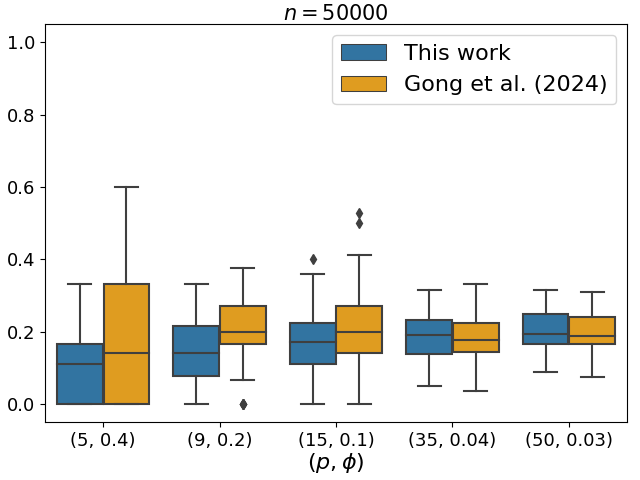}
    \end{subfigure}
     \caption{Box plots of $\mbox{UNED}(\mathcal{E}, \hat{\mathcal{E}})$ for $n = 5000$ (left) and $n = 50000$ (right) samples under different combinations of scale and connectivity parameters $(p, \phi)$, compared with the method from~\cite{gong2024partial}. Each experiment repeated for $M = 50$ randomly initialized extremal Markov networks. The parameters for~\cite{gong2024partial} are optimally tuned via grid search; our method uses fixed parameters as described in Section~\ref{sec:sim:overview}.}
    \label{fig:extremalmarkov:performance}
\end{figure}

\subsection{Estimating DAGs for XSCMs from time series data}\label{sec:sim:XSCM:time-series}

Initialization of the XSCM for time series follows the same procedure as with the cross-sectional case. To ensure acyclicity within each time slice, $\bm{B}_{(0)}$ is constrained to be lower triangular, while the lagged matrices $\bm{B}_{(1)}, \dots, \bm{B}_{(\tau)}$ have no such constraints. We control the overall sparsity level using a connectivity parameter $\phi\in(0,1)$, and draw non-zero entries uniformly on $[0, 1]$.
The corresponding graph $\mathcal{G} = (\mathcal{V}, \mathcal{E})$ is constructed with $\mathcal{V} = \{v_{1,t}, \dots, v_{p,t}\}_{t=1}^T$ and ${\mathcal{E} = \{(v_{j,t-\delta}, v_{i,t}) : (B_{(\delta);ij} > 0\}}$. Then, the time series $\{\bm{X}_t\}_{t=1}^T$ is generated iteratively using \eqref{eq:XSCM:timeSeries:direct}. See Appendix~\ref{sec:appendix:simulation:DAGtime}, for an example.

We consider a maximum lag $\tau = 1$ and simulate $T = 5000$ time points for $M=50$ different XSCMs time series. As the method proposed by~\cite{bodik2024granger} cannot capture contemporaneous causal effects, we conduct two sets of experiments. In the first, we impose $\bm{B}_{(0)} \equiv \bm{0}$ (the zero matrix) to remove contemporaneous effects. In the second, $\bm{B}_{(0)}$ is allowed to contain non-zero entries. We apply both our method and the method of~\cite{bodik2024granger} to estimate the underlying DAGs, evaluating performance using the normalized edit distance $\mbox{NED}(\mathcal{E}, \hat{\mathcal{E}})$ defined in Equation~\eqref{eq:edit}, while ignoring edges associated with contemporaneous effects.

Results under varying combinations of scale and connectivity parameters $(p, \phi)$ are summarized in Figure~\ref{fig:oneTail}. The left panel shows the method's performance under $\bm{B}_{(0)} \equiv \bm{0}$, and the right panel shows the performance when contemporaneous effects are accounted for. When $\bm{B}_{(0)}$ is zero, our method performs better than~\cite{bodik2024granger} but with a small margin. However, once contemporaneous effects are introduced, the performance of the method in~\cite{bodik2024granger} deteriorates significantly, while our method remains robust. The performance gap increases with the number of variables, $p$.

\begin{figure}
    \centering
    % 左边的子图
    \begin{subfigure}{0.475\textwidth}
        \centering
        \includegraphics[width=\textwidth]{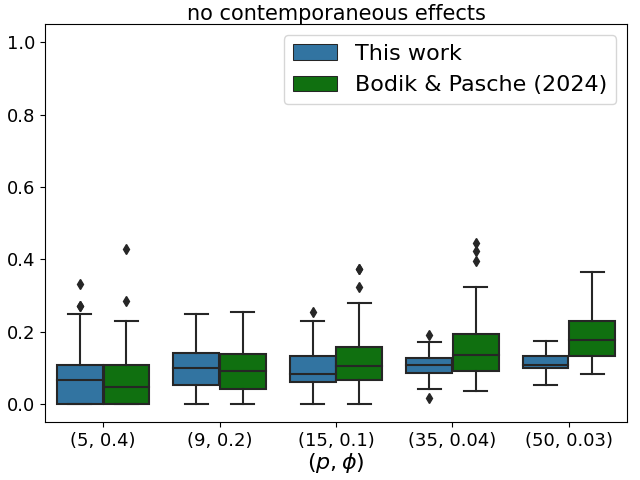}
    \end{subfigure}
    \hspace{0.5cm}
    % 右边的子图
    \begin{subfigure}{0.475\textwidth}
        \includegraphics[width=\textwidth]{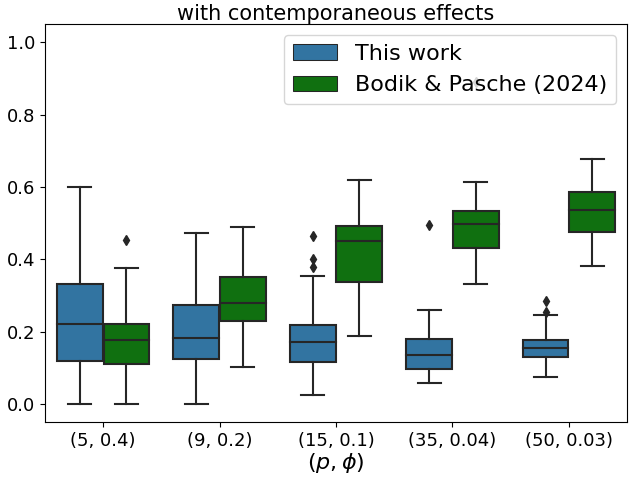} 
    \end{subfigure}

    \caption{Box plots of error measured by $\mbox{NED}$ (contemporaneous edges ignored) under different combinations of scale and sparsity parameters $(p, \phi)$
    with (right) and without ($\bm{B}_{(0)} \equiv \bm{0}$, left) contemporaneous effects. Each experiment is repeated for $M = 50$ randomly initialized XSCM time series of length $T = 5000$.}
    \label{fig:oneTail}
\end{figure}

\section{Real data applications}
\label{sec:app}

This section presents two applications of our proposed method. Section~\ref{sec:app:river} demonstrates the estimation of the underlying DAG structure in the Danube river basin dataset, a well-known benchmark application in hydrology and widely used in the literature on graphical modeling for extremes~\citep[see, e.g.,][]{asadi2015extremes,engelke2020graphical,kim2022hypothesis,gong2024partial,bodik2024granger}. Section~\ref{sec:app:china} investigates causal discovery in trading activities of China's future market, based on a high-frequency financial dataset constructed by~\cite{jiang2024efficient}.

\subsection{Danube river basin data}
\label{sec:app:river}

The Danube river basin dataset contains daily discharge data from 31 gauging stations along the Danube river basin spanning the period from 1960 to 2009. We use the preprocessed version provided by~\cite{asadi2015extremes}, which includes only summer months (June, July, August), resulting in a total of \NumObsRiver\ observations. The known river topology is shown in the upper panel of Figure~\ref{fig:manyResultsOnDanube}, where edge directions indicate flow from upstream to downstream.

\begin{figure}
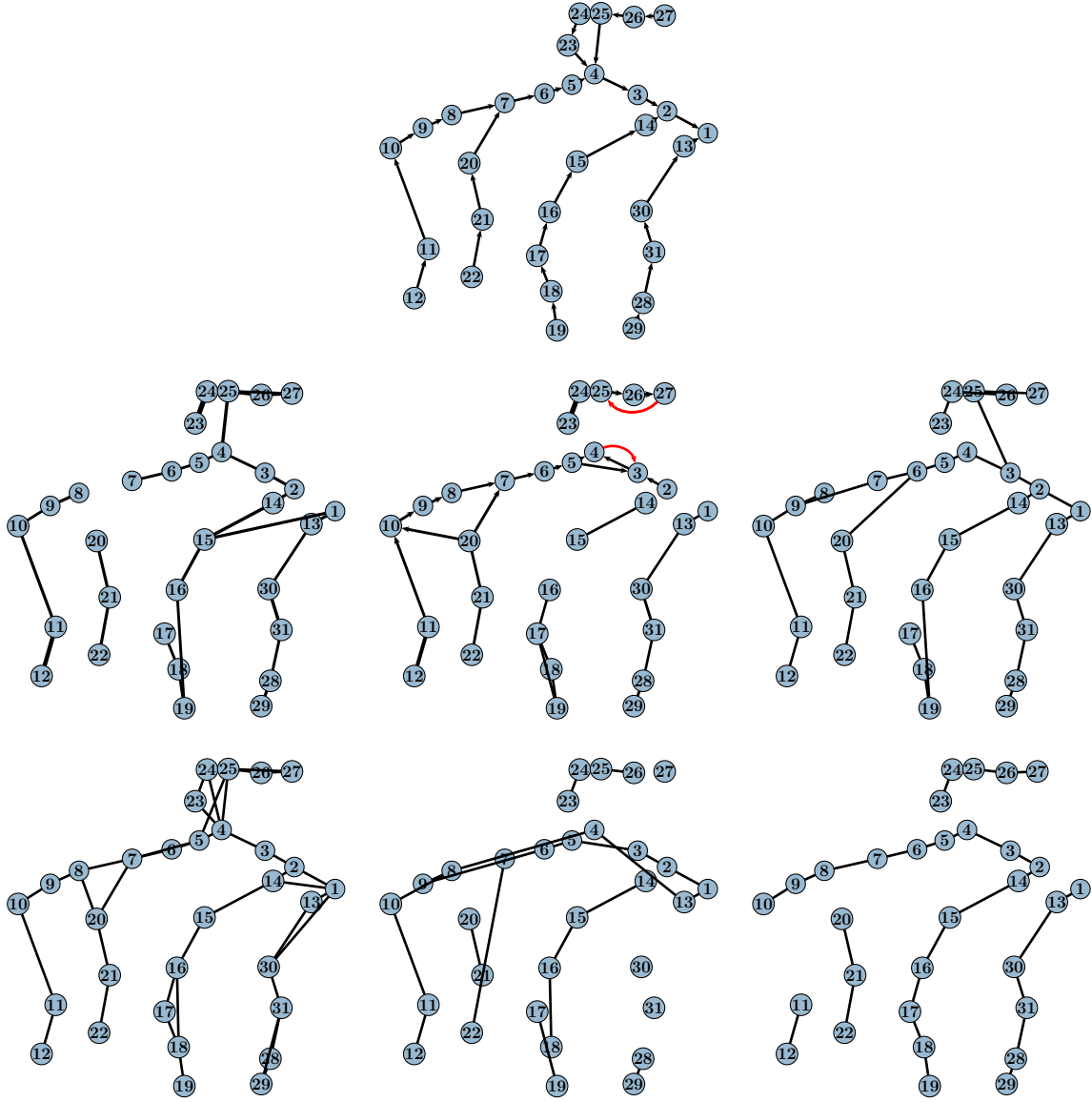

    \centering
    \begin{subfigure}[b]{0.31\textwidth}
        \begin{tikzpicture}[scale=0.22,transform shape,every node/.style={circle, draw, outer sep=0pt,minimum size=35pt,fill=danubeBlue!50, inner sep=0pt,font=\bfseries\fontsize{30}{35}\selectfont
}]
            \input{pic/ground_truth.tex}
                \end{tikzpicture}
    \end{subfigure}

    \vspace{1.2em}
    
    \begin{subfigure}[b]{0.31\textwidth}
        \begin{tikzpicture}[scale=0.22,transform shape,every node/.style={circle, draw, outer sep=0pt,minimum size=35pt,fill=danubeBlue!50, inner sep=0pt,font=\bfseries\fontsize{30}{35}\selectfont
}]
            \input{pic/oursT0.tex}
                \end{tikzpicture}
    \end{subfigure}
    \begin{subfigure}[b]{0.31\textwidth}
        \begin{tikzpicture}[scale=0.22,transform shape,every node/.style={circle, draw, outer sep=0pt,minimum size=35pt,fill=danubeBlue!50, inner sep=0pt,font=\bfseries\fontsize{30}{35}\selectfont
}]
            \input{pic/oursT1.tex}
                \end{tikzpicture}
    \end{subfigure}
    \begin{subfigure}[b]{0.31\textwidth}
        \begin{tikzpicture}[scale=0.22,transform shape,every node/.style={circle, draw, outer sep=0pt,minimum size=35pt,fill=danubeBlue!50, inner sep=0pt,font=\bfseries\fontsize{30}{35}\selectfont
}]
        \input{pic/enge.tex}
        \end{tikzpicture}
    \end{subfigure}

    \vspace{1.2em}
    \begin{subfigure}[b]{0.31\textwidth}
        \begin{tikzpicture}[scale=0.22,transform shape,every node/.style={circle, draw, outer sep=0pt,minimum size=35pt,fill=danubeBlue!50, inner sep=0pt,font=\bfseries\fontsize{30}{35}\selectfont
}]
        \input{pic/engeGraphBased.tex}
        \end{tikzpicture}
    \end{subfigure}
    \begin{subfigure}[b]{0.31\textwidth}
        \begin{tikzpicture}[scale=0.22,transform shape,every node/.style={circle, draw, outer sep=0pt,minimum size=35pt,fill=danubeBlue!50, inner sep=0pt,font=\bfseries\fontsize{30}{35}\selectfont
}]
            \input{pic/lee.tex}
            \end{tikzpicture}
    \end{subfigure}
    \begin{subfigure}[b]{0.31\textwidth}
        \begin{tikzpicture}[scale=0.22,transform shape,every node/.style={circle, draw, outer sep=0pt,minimum size=35pt,fill=danubeBlue!50, inner sep=0pt,font=\bfseries\fontsize{30}{35}\selectfont
}]
        \input{pic/gong.tex}
        \end{tikzpicture}
    \end{subfigure}  
    \caption{Known river topology (top); estimation by our method ignoring lagged causality ($\tau = 0$, middle-left) with edge thickness proportional to PTCC scores; estimation by our method incorporating one-step lagged causality ($\tau = 1$, middle-middle) with edge thickness proportional to PTCC scores and lagged causal effects in red; estimation by~\cite{engelke2020graphical} when assuming the underlying structure to be a tree (middle-right); estimation by~\cite{engelke2020graphical} when assuming the underlying structure to be a undirected graph (bottom-left); estimation by~\cite{kim2022hypothesis} (bottom-middle) and~\cite{gong2024partial} (bottom-right).}
    \label{fig:manyResultsOnDanube}
\end{figure}

To satisfy the regular variation condition in Definition~\ref{def:rv}, we first apply an empirical rank transformation to each margin as described in Section~\ref{sec:bg:RegularVarying}, resulting in unit Pareto margins with shape parameter 2. We then apply our method to estimate the underlying DAG. The two free hyperparameters are set to $\alpha = 0.005,q = 0.99$. The estimated graphs for maximum causal lags $\tau = 0$ and $\tau = 1$ are shown in the middle-left and middle-middle panels of Figure~\ref{fig:manyResultsOnDanube}, respectively. In those two graphs, edge thickness is proportional to the PTCC scores. Directed edges indicate identifiable causal directions; undirected edges represent dependencies with unidentifiable directions. Red edges denote one-step lagged causal effects. A sensitivity analysis with respect to the hyperparameters $\alpha$ and $q$ is provided in Appendix~\ref{sec:danubeSensitivity}, where we repeat the experiment over a wide range of values for $\alpha$ and $q$. The results remain consistent across different settings, showing that the underlying skeleton can be largely recovered. These observations confirm that our method is robust to the choice of hyperparameters.

When lagged effects are excluded ($\tau = 0$, middle-left), our method successfully recovers most of the skeleton of the river structure, although many contemporaneous directions remain unidentified. When incorporating one-step lagged causality ($\tau = 1$, middle-middle), directionality estimation improves. For example, a collider structure is detected at station~7. Additionally, one-step lagged causal connections are identified between stations 3 and 4, and between stations 25 and 27—both of which are consistent with the known river topology.

To facilitate comparison, we also present results from several existing methods applied to the same dataset. In~\cite{engelke2020graphical}, results assuming that the underlying structure is a tree or an undirected graph are shown in the middle-right and bottom-left panels of Figure~\ref{fig:manyResultsOnDanube}, respectively. The former achieves better performance than the latter, as it is a stronger assumption than assuming that the underlying structure is a general undirected graph. The results by~\cite{kim2022hypothesis} and~\cite{gong2024partial} are shown in the bottom-middle and bottom-right panels, respectively. Overall, our method offers greater expressiveness and interpretability by learning the DAG structure, and it is more general, being applicable to time series data.

\subsection{Tail causality in China's future market}
\label{sec:app:china}

We now also apply our method to high-frequency trading data from China's future market, which constitutes a subset of the time series dataset constructed by~\cite{jiang2024efficient}. Our analysis focuses on tail causal discovery in trading activities during the period from \TrainStartDate\ to \TrainEndDate. The dataset covers \NumberOfProduct\ products across various categories, from traditional commodities such as grains and coal to financial instruments like equity index and interest rate futures, as listed in Table~\ref{tab:productCategories} with their categories. For details on the specific assets represented by each product code, please refer to the~\cite{jiang2024efficient}, Section~B of the Supplementary Material.  Trading activity is measured as the aggregated trading volume over 5-minute intervals, normalized by the total daily volume, yielding \NumObsChina\ observations.

We apply the marginal transformation described in Section~\ref{sec:app:river} to satisfy the regular variation condition. Our method is then applied fixing the hyperparameters as $\alpha = 0.005$, $q = 0.99,$ and maximum causal lag as $\tau = 1$. The resulting causal graph is shown in Figure~\ref{fig:clusteringBasedWhole}, where the edge widths are proportional to the PTCC scores. To better visualize the structure, we cluster the assets based on the PTCC scores for contemporaneous relationships using the $k$-means algorithm. The number of clusters is set to $k=18$, corresponding to the number of categories in Table~\ref{tab:productCategories}.

%with corresponding categories\href{https://www.ceicdata.com/en/indicator/china/short-term-interest-rate}{CEIC}. 
\begin{table}
    \centering
    \caption{Futures categories and corresponding asset codes.}
    \vspace{0.5em}
    \resizebox{0.94\textwidth}{!}{
    \begin{tabular}{@{}ll|ll@{}}
    \toprule
    \textbf{Category} & \textbf{Codes}            & \textbf{Category} & \textbf{Codes}        \\
    \midrule
    Oil crops        & a, m, OI, p, b, RM, y     & Precious metals   & ag, au                \\
    Nonferrous metals & al, bc, cu, ni, pb, sn, zn, ao & Economic crops & AP, CF, CJ, CY, PK, SR \\
    Rubber \& woods  & br, fb, nr, ru, sp        & Oil \& gas        & bu, fu, lu, pg, sc     \\
    Grains           & c, cs                     & Olefins           & eb, l, pp, v           \\
    Alcohols         & eg, MA                    & Inorganics        & FG, SA, UR, SH         \\
    Ferrous metals   & hc, i, rb, SF, SM, ss     & Equity index      & IC, IF, IH, IM         \\
    Coals            & j, jm                     & Animals           & jd, lh                 \\
    Novel materials  & lc, si                    & Aromatics         & PF, TA, PX             \\
    Interest rates   & T, TF, TL, TS             & Indices           & ec                     \\
    \bottomrule
    \end{tabular}
    }
    \label{tab:productCategories}
\end{table}

\begin{figure}
    \centering
    \includegraphics[width=1\linewidth]{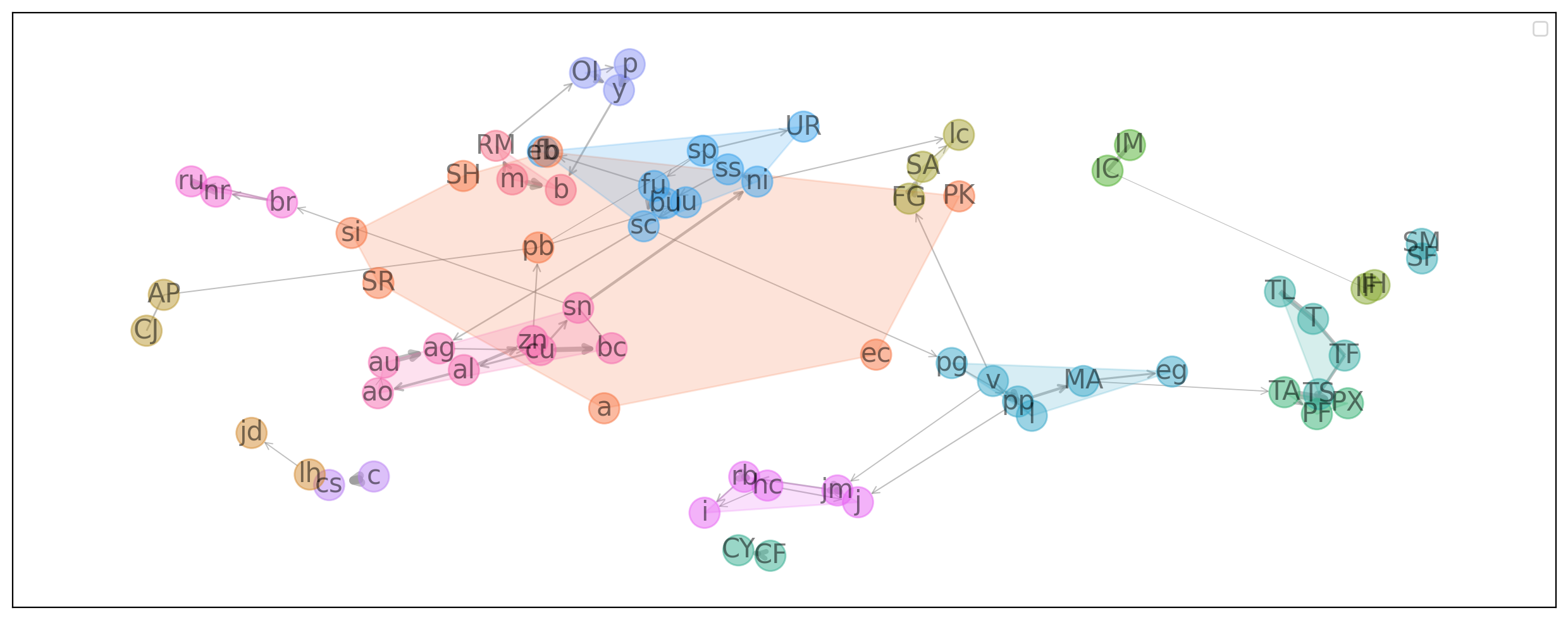}
    \caption{Estimated causal structure of market-wide trading activity in China's futures market (2024--2025). Edge width is proportional to PTCC scores. Clustering is based on the contemporaneous adjacency matrix using the $k$-means algorithm with $k=18$ clusters.}
    \label{fig:clusteringBasedWhole}
\end{figure}

The results indicate that large trading activity tends to be more strongly associated with large trading activity within the same product category. For example, the national bond futures---T, TL, TS, and TF---exhibit strong contemporaneous connections. Similar patterns are observed in agricultural products (OI, p, y, RM, m, b) and chemical products (TA, PF, PX). Furthermore, thanks to the collider structure identified by our method, partial directionality can be inferred even in contemporaneous relationships. Examples include lh$\rightarrow$jd, au$\rightarrow$ag, among others. We find that the direction of edges tend to originate from assets with higher trading activity. Specifically, by comparing the daily average trading turnover (measured in USD) for each pair of directed edges, we observe that in $36$ out of the $56$ pairs, the asset identified as the cause has a larger trading turnover than the effect asset.

\section{Conclusion}
\label{sec:conclusion}

We have introduced a new structural causal model—the XSCM—to capture causal direction in tail-dependent (and potentially heavy-tailed) data. Our main theoretical contribution is the establishment of an equivalence between zero PTCC and graph separation, applicable to both XSCMs and extremal Markov networks. This equivalence allows the problem of learning both DAG structures (in XSCMs) and undirected structures (in extremal Markov networks) to be reformulated as separation detection in graphs, making the use of popular constraint-based learning algorithms applicable.

To evaluate the efficiency of our method, we conducted simulation studies in three distinct settings: inferring DAGs from cross-sectional data, inferring DAGs from time series data, and estimating undirected graph structures. In all cases, our method achieves strong performance compared to existing approaches.
We further demonstrated the practical value of our method through two real-world applications. For the Danube river basin dataset, our method successfully recovers the known river topology and identifies collider structures and lagged causal relationships. In the high-frequency trading data application from China's futures market, the learned graph structure aligns well with product categories, capturing both within-category dependence and partial directionality in extremal trading activity.

In parallel work, the recent preprint of \citet{engelke2025extremes} derived an extremal SCM by analyzing the limiting tail behavior of classical SCMs. Although both works independently tackle causality in extremes through SCMs, the two proposed frameworks remain distinct, and it would be interesting in future research to study their respective strengths and limitations. One of the key differences is that our approach directly constructs a new causal framework tailored to extremes, rather than deriving it from asymptotic arguments, and thus it provides a foundation for structure learning in this setting. Another possible direction for future research is to study the causal mechanisms of extreme events that occur in both directions of a random variable. Modeling such mechanisms—where extremes may arise in both the upper and lower tails—is receiving increasing attention in the literature; see, for example,~\cite{gnecco2021causal} and \cite{bodik2024granger}, where the variable of interest is transformed by taking its absolute value so that existing methods can be applied. However, given the growing acknowledgment of the differing behaviors in the joint upper and lower tail regions, as well as evidence of complex cross-directional dependence \citep{jiang2024efficient}, explicitly modeling the interaction between the upper and lower tails is increasingly necessary. This challenge calls for the development of new structural causal models that account for the inherent constraint between the two tails of the same variable, and may open up a promising line of research.
\section*{Competing interests}
No competing interest is declared.

\section*{Data availability}
The authors confirm that the data supporting the findings of this study are available within the supplementary materials.
\if1\anon
\else
\section*{Acknowledgments}
The author gratefully acknowledge valuable discussions with Dr. Nabila Bounceur on state-of-the-art causal discovery methods.
\fi

{\renewcommand{\baselinestretch}{0.80}\normalsize
\bibliographystyle{apalike}
\bibliography{reference}

@article{engelke2021learning,
  title={Learning extremal graphical structures in high dimensions},
  author={Engelke, Sebastian and Lalancette, Micha{\"e}l and Volgushev, Stanislav},
  journal={arXiv preprint arXiv:2111.00840},
  year={2021}
}

@article{lederer2023extremes,
  title={Extremes in high dimensions: Methods and scalable algorithms},
  author={Lederer, Johannes and Oesting, Marco},
  journal={arXiv preprint arXiv:2303.04258},
  year={2023}
}

@article{tran2024estimating,
  title={Estimating a directed tree for extremes},
  author={Tran, Ngoc Mai and Buck, Johannes and Kl{\"u}ppelberg, Claudia},
  journal={Journal of the Royal Statistical Society Series B: Statistical Methodology},
  volume={86},
  number={3},
  pages={771--792},
  year={2024},
  publisher={Oxford University Press UK}
}

@article{wan2023graphical,
  title={Graphical lasso for extremes},
  author={Wan, Phyllis and Zhou, Chen},
  journal={arXiv preprint arXiv:2307.15004},
  year={2023}
}

@article{paluvs2024causes,
  title={Causes of extreme events revealed by R{\'e}nyi information transfer},
  author={Palu{\v{s}}, Milan and Chvostekov{\'a}, Martina and Manshour, Pouya},
  journal={Science Advances},
  volume={10},
  number={30},
  year={2024},
  publisher={American Association for the Advancement of Science}
}

@article{hill1975simple,
  title={A simple general approach to inference about the tail of a distribution},
  author={Hill, Bruce M},
  journal={The Annals of Statistics},
  pages={1163--1174},
  year={1975},
  publisher={JSTOR}
}

@article{engelke2024graphical,
  title={Graphical models for multivariate extremes},
  author={Engelke, Sebastian and Hentschel, Manuel and Lalancette, Micha{\"e}l and R{\"o}ttger, Frank},
  journal={arXiv preprint arXiv:2402.02187},
  year={2024}
}

@article{colombo2012learning,
  title={Learning high-dimensional directed acyclic graphs with latent and selection variables},
  author={Colombo, Diego and Maathuis, Marloes H and Kalisch, Markus and Richardson, Thomas S},
  journal={The Annals of Statistics},
  pages={294--321},
  volume={40},
  number={1},
  year={2012},
  publisher={JSTOR}
}

@book{pearl2009causality,
  title     = {Causality: Models, Reasoning, and Inference},
  author    = {Pearl, Judea},
  year      = {2009},
  edition   = {2nd},
  publisher = {Cambridge University Press},
  address   = {Cambridge, UK}
}

@book{pearl2014probabilistic,
  title={Probabilistic Reasoning in Intelligent Systems: Networks of Plausible Inference},
  author={Pearl, Judea},
  year={2014},
  publisher={Elsevier}
}

@article{krali2023heavy,
  title={Heavy-tailed max-linear structural equation models in networks with hidden nodes},
  author={Krali, Mario and Davison, Anthony C and Kl{\"u}ppelberg, Claudia},
  journal={arXiv preprint arXiv:2306.15356},
  year={2023}
}

@article{pasche_causal_2023,
	title = {Causal modelling of heavy-tailed variables and confounders with application to river flow},
	volume = {26},
	language = {en},
	number = {3},
	urldate = {2024-09-30},
	journal = {Extremes},
	author = {Pasche, Olivier C. and Chavez-Demoulin, Valérie and Davison, Anthony C.},
	year = {2023},
	pages = {573--594}
}

@book{resnick2007heavy,
  title={Heavy-tail Phenomena: Probabilistic and Statistical Modeling},
  author={Resnick, Sidney I},
  year={2007},
  publisher={Springer Science \& Business Media}
}

@article{kim2022hypothesis,
  title={Hypothesis testing for partial tail correlation in multivariate extremes},
  author={Kim, Mihyun and Lee, Jeongjin},
  journal={arXiv e-prints},
  pages={arXiv--2210},
  year={2022}
}

@inproceedings{runge2020discovering,
  title={Discovering contemporaneous and lagged causal relations in autocorrelated nonlinear time series datasets},
  author={Runge, Jakob},
  booktitle={Conference on Uncertainty in Artificial Intelligence},
  pages={1388--1397},
  year={2020},
  organization={PMLR}
}

@book{rue2005gaussian,
  title={Gaussian Markov Random Fields: Theory and Applications},
  author={Rue, Havard and Held, Leonhard},
  year={2005},
  publisher={Chapman and Hall/CRC}
}

@article{runge2018causal,
  title={Causal network reconstruction from time series: From theoretical assumptions to practical estimation},
  author={Runge, Jakob},
  journal={Chaos: An Interdisciplinary Journal of Nonlinear Science},
  volume={28},
  number={7},
  pages={075310},
  year={2018},
  publisher={AIP Publishing}
}

@article{hammoudeh2020relationship,
  title={Relationship between green bonds and financial and environmental variables: A novel time-varying causality},
  author={Hammoudeh, Shawkat and Ajmi, Ahdi Noomen and Mokni, Khaled},
  journal={Energy Economics},
  volume={92},
  pages={104941},
  year={2020},
  publisher={Elsevier}
}

@article{krali2025causal,
  title={Causal discovery in heavy-tailed linear structural equation models via scalings},
  author={Krali, Mario},
  journal={Scandinavian Journal of Statistics},
  volume={53},
  number={1},
  pages={291--334},
  year={2026},
  publisher={Wiley Online Library}
}

@article{shojaie2022granger,
  title={Granger causality: A review and recent advances},
  author={Shojaie, Ali and Fox, Emily B},
  journal={Annual Review of Statistics and Its Application},
  volume={9},
  pages={289--319},
  year={2022},
  publisher={Annual Reviews}
}

@article{larsson2012extremal,
  title={Extremal dependence measure and extremogram: the regularly varying case},
  author={Larsson, Martin and Resnick, Sidney I},
  journal={Extremes},
  volume={15},
  number={2},
  pages={231--256},
  year={2012},
  publisher={Springer}
}

@article{cooley2019decompositions,
  title={Decompositions of dependence for high-dimensional extremes},
  author={Cooley, Daniel and Thibaud, Emeric},
  journal={Biometrika},
  volume={106},
  number={3},
  pages={587--604},
  year={2019},
  publisher={Oxford University Press}
}

@article{engelke2020graphical,
  title={Graphical models for extremes},
  author={Engelke, Sebastian and Hitz, Adrien S},
  journal={Journal of the Royal Statistical Society Series B: Statistical Methodology},
  volume={82},
  number={4},
  pages={871--932},
  year={2020},
  publisher={Oxford University Press}
}

@article{asadi2015extremes,
  title={Extremes on river networks},
  author={Asadi, Peiman and Davison, Anthony C and Engelke, Sebastian},
  journal={The Annals of Applied Statistics},
  volume={9},
  number={4},
  pages={2023–2050},
  year={2015},
  publisher={Institute of Mathematical Statistics}
}

@book{spirtes2000causation,
  title={Causation, Prediction, and Search},
  author={Spirtes, Peter and Glymour, Clark N and Scheines, Richard},
  year={2000},
  publisher={MIT press}
}

@article{runge2015identifying,
  title={Identifying causal gateways and mediators in complex spatio-temporal systems},
  author={Runge, Jakob and Petoukhov, Vladimir and Donges, Jonathan F and Hlinka, Jaroslav and Jajcay, Nikola and Vejmelka, Martin and Hartman, David and Marwan, Norbert and Palu{\v{s}}, Milan and Kurths, J{\"u}rgen},
  journal={Nature Communications},
  volume={6},
  number={1},
  pages={8502},
  year={2015},
  publisher={Nature Publishing Group UK London}
}

@article{engelke2025extremes,
  title={Extremes of structural causal models},
  author={Engelke, Sebastian and Gnecco, Nicola and R{\"o}ttger, Frank},
  journal={arXiv preprint arXiv:2503.06536},
  year={2025}
}

@article{gong2024partial,
  title={Partial Tail-Correlation Coefficient Applied to Extremal-Network Learning},
  author={Gong, Yan and Zhong, Peng and Opitz, Thomas and Huser, Rapha{\"e}l},
  journal={Technometrics},
  volume={66},
  number={3},
  pages={331--346},
  year={2024},
  publisher={Taylor \& Francis}
}

@article{lee2021transformed,
  title={Transformed-linear prediction for extremes},
  author={Lee, Jeongjin and Cooley, Daniel},
  journal={arXiv preprint arXiv:2111.03754},
  year={2021}
}

@article{gnecco2021causal,
  title={Causal discovery in heavy-tailed models},
  author={Gnecco, Nicola and Meinshausen, Nicolai and Peters, Jonas and Engelke, Sebastian},
  journal={The Annals of Statistics},
  volume={49},
  number={3},
  pages={1755--1778},
  year={2021},
  publisher={Institute of Mathematical Statistics}
}

@article{kluppelberg2021estimating,
  title={Estimating an extreme {B}ayesian network via scalings},
  author={Kl{\"u}ppelberg, Claudia and Krali, Mario},
  journal={Journal of Multivariate Analysis},
  volume={181},
  pages={104672},
  year={2021},
  publisher={Elsevier}
}

@article{engelke2022structure,
  title={Structure learning for extremal tree models},
  author={Engelke, Sebastian and Volgushev, Stanislav},
  journal={Journal of the Royal Statistical Society: Series B (Statistical Methodology)},
  volume={84},
  number={5},
  pages={2055--2087},
  year={2022},
  publisher={Wiley Online Library}
}

@article{engelke2021sparse,
  title={Sparse structures for multivariate extremes},
  author={Engelke, Sebastian and Ivanovs, Jevgenijs},
  journal={Annual Review of Statistics and Its Application},
  volume={8},
  number={1},
  pages={241--270},
  year={2021},
  publisher={Annual Reviews}
}

@article{imai2013experimental,
  title={Experimental designs for identifying causal mechanisms},
  author={Imai, Kosuke and Tingley, Dustin and Yamamoto, Teppei},
  journal={Journal of the Royal Statistical Society Series A: Statistics in Society},
  volume={176},
  number={1},
  pages={5--51},
  year={2013},
  publisher={Oxford University Press}
}

@inproceedings{spirtes2001anytime,
  title={An anytime algorithm for causal inference},
  author={Spirtes, Peter},
  booktitle={International Workshop on Artificial Intelligence and Statistics},
  pages={278--285},
  year={2001},
  organization={PMLR}
}

@book{peters2017elements,
  title={Elements of Causal Inference: Foundations and Learning Algorithms},
  author={Peters, Jonas and Janzing, Dominik and Sch{\"o}lkopf, Bernhard},
  year={2017},
  publisher={The MIT Press}
}

@article{bodik2024causality,
  title={Causality in extremes of time series},
  author={Bodik, Juraj and Palu{\v{s}}, Milan and Pawlas, Zbyn{\v{e}}k},
  journal={Extremes},
  volume={27},
  number={1},
  pages={67--121},
  year={2024},
  publisher={Springer}
}

@article{bodik2024granger,
  title={Granger Causality in Extremes},
  author={Bodik, Juraj and Pasche, Olivier},
  journal={arXiv preprint arXiv:2407.09632},
  year={2024}
}

@article{chavez2024causality,
  title={Causality and extremes},
  author={Chavez-Demoulin, Val{\'e}rie and Mhalla, Linda},
  journal={arXiv preprint arXiv:2403.05331},
  year={2024}
}

@article{jiang2024efficient,
  title={The efficient tail hypothesis: An extreme value perspective on market efficiency},
  author={Jiang, Junshu and Richards, Jordan and Huser, Rapha{\"e}l and Bolin, David},
  journal={Journal of Business \& Economic Statistics},
  pages={1--14},
  year={2025},
  publisher={Taylor \& Francis}
}

@book{Berman_1994, 
location={Philadelphia, Pa.}, 
title={Nonnegative Matrices in the Mathematical Sciences}, 
ISBN={9780898713213}, 
publisher={Society for Industrial and Applied Mathematics}, 
author={Berman, Abraham}, 
year={1994}
}

@book{Isaacson_1994,
location={Mineola, NY.},
title={Analysis of Numerical Methods},
ISBN={9780486680293},
publisher={Dover Publications},
author={Isaacson, Eugene}, 
year={1994}
}

@book{Cormen_Leiserson_2022, 
  location={London, England},
  title={Introduction to Algorithms}, 
  edition = {Fourth},
  ISBN={9780262046305}, 
  publisher={MIT Press}, 
  author={Cormen, Thomas H. and Leiserson, Charles E.},
  year={2022}
}
}

\newpage

\begin{appendices}

\section{More on simulation}\label{sec:appendix:simulation}

\subsection{Estimating DAGs for XSCMs with cross-sectional data}
\label{sec:appendix:simulation:DAGnontime}

Figure~\ref{fig:XSCM:scatter} shows the pairwise scatter plots for $n = 5000$ samples generated from the XSCM shown in Figure~\ref{fig:XSCM:illustration}.
\begin{figure}[H]
    \centering
    \includegraphics[width=0.95\textwidth]{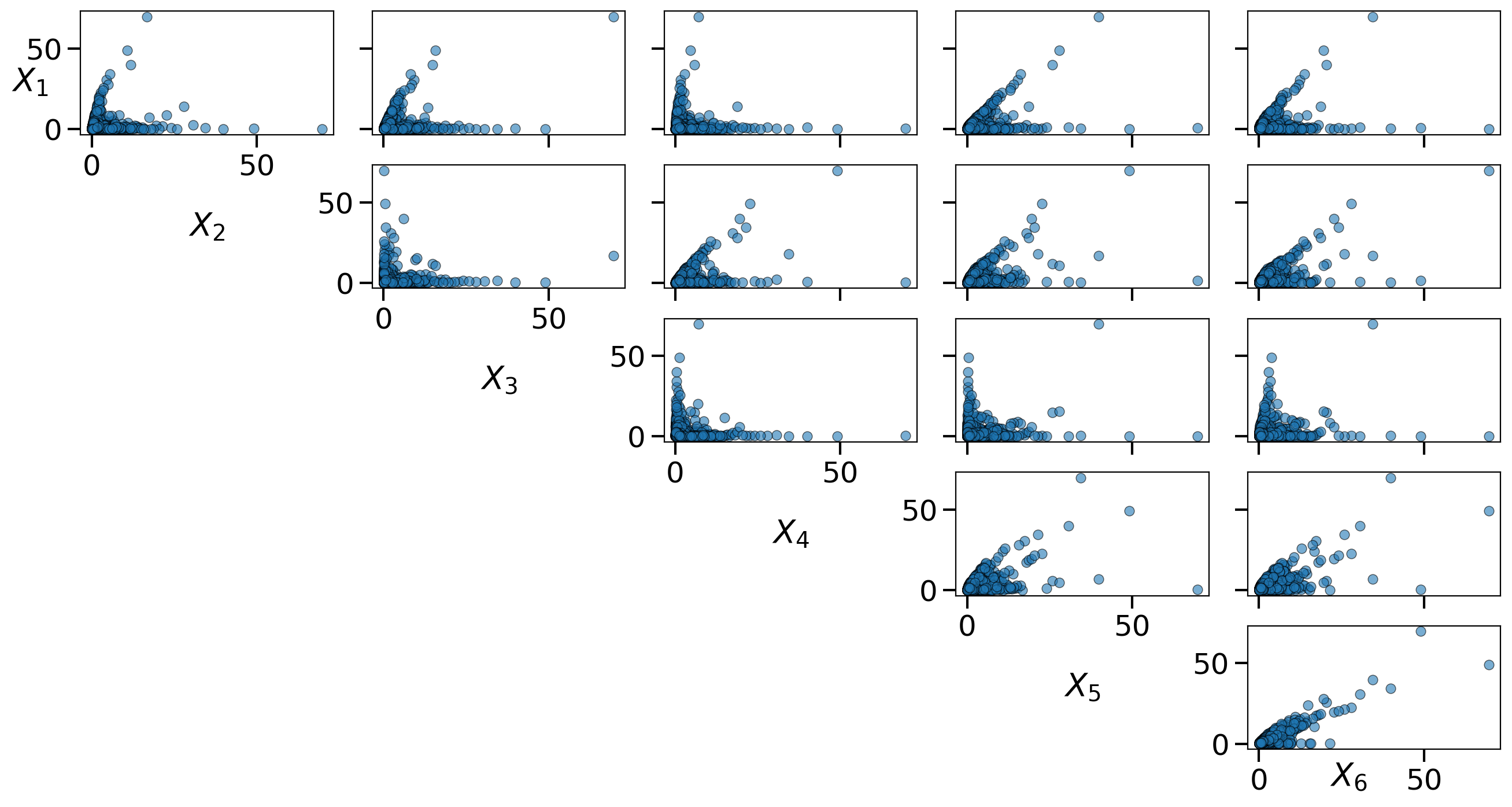}
    \caption{Pairwise scatter plots of the six components of $(X_1,\dots,X_6)^\top$ generated from the example XSCM in Figure~\ref{fig:XSCM:illustration}. All variables are marginally transformed to have a Pareto distribution with unit scale and shape parameter 2.} \label{fig:XSCM:scatter} 
    \end{figure}

\subsection{Estimating the undirected graph of extremal Markov network}
\label{sec:appendix:simulation:emn}

Figure~\ref{fig:extremalmarkov:illustration} shows an example of a randomly initialized precision matrix $\bm{Q}$ (left panel) and its corresponding undirected graph structure (right panel). Figure~\ref{fig:extremalmarkov:scatter} presents the pairwise scatter plots of $n = 5000$ samples generated from the extremal Markov network in Figure~\ref{fig:extremalmarkov:illustration}.

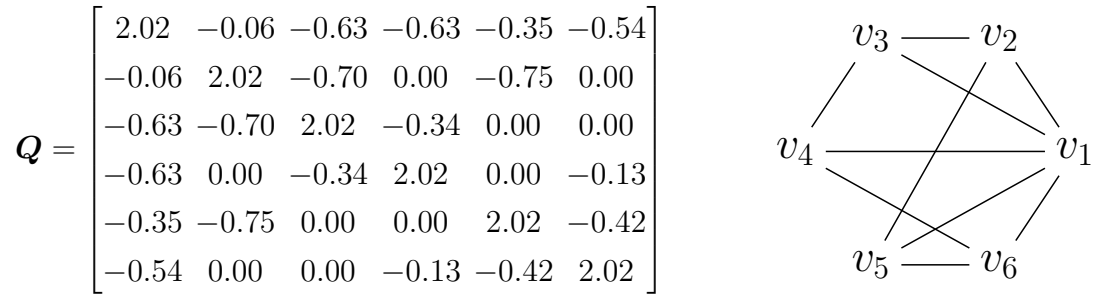
\begin{figure}[H]
    \centering
    % 第一个子图和第一个矩阵
    \begin{minipage}{0.48\textwidth}
        \vspace{-0.5cm}
        \centering
\[
\renewcommand{\arraystretch}{0.7} 
\bm{Q} =
\left[
\begin{array}{@{}c@{\hspace{0.4em}}c@{\hspace{0.4em}}c@{\hspace{0.4em}}c@{\hspace{0.4em}}c@{\hspace{0.4em}}c@{}}
2.02  & -0.06 & -0.63 & -0.63 & -0.35 & -0.54 \\
-0.06 &  2.02 & -0.70 &  0.00 & -0.75 &  0.00 \\
-0.63 & -0.70 &  2.02 & -0.34 &  0.00 &  0.00 \\
-0.63 &  0.00 & -0.34 &  2.02 &  0.00 & -0.13 \\
-0.35 & -0.75 &  0.00 &  0.00 &  2.02 & -0.42 \\
-0.54 &  0.00 &  0.00 & -0.13 & -0.42 &  2.02
\end{array}
\right]
\]
    \end{minipage}
    \hspace{0.2cm}
    \begin{minipage}{0.48\textwidth}
        \centering
        \begin{tikzpicture}
            \node (v1) at (3.7,1.5)  {\Large $v_1$};
            \node (v2) at (2.7,3)  {\Large $v_2$};
            \node (v3) at (1,3)    {\Large $v_3$};
            \node (v4) at (0,1.5)    {\Large $v_4$};
            \node (v5) at (1,0)   {\Large $v_5$};
            \node (v6) at (2.7,0) {\Large $v_6$};            
            \draw[-, line width=0.2mm, >=stealth] (v1) -- (v2); 
            \draw[-, line width=0.2mm, >=stealth] (v1) -- (v3); 
            \draw[-, line width=0.2mm, >=stealth] (v1) -- (v4); 
            \draw[-, line width=0.2mm, >=stealth] (v1) -- (v5); 
            \draw[-, line width=0.2mm, >=stealth] (v1) -- (v6); 
            \draw[-, line width=0.2mm, >=stealth] (v2) -- (v3); 
            \draw[-, line width=0.2mm, >=stealth] (v2) -- (v5); 
            \draw[-, line width=0.2mm, >=stealth] (v3) -- (v4); 
            \draw[-, line width=0.2mm, >=stealth] (v4) -- (v6); 
            \draw[-, line width=0.2mm, >=stealth] (v5) -- (v6); 
        \end{tikzpicture}
    \end{minipage} 
\caption{A randomly initialized precision matrix $\bm{Q}$ (left) and its corresponding undirected graph structure (right).} 
\label{fig:extremalmarkov:illustration}
\end{figure}

\begin{figure}[H]
    \centering
    \includegraphics[width=0.95\textwidth]{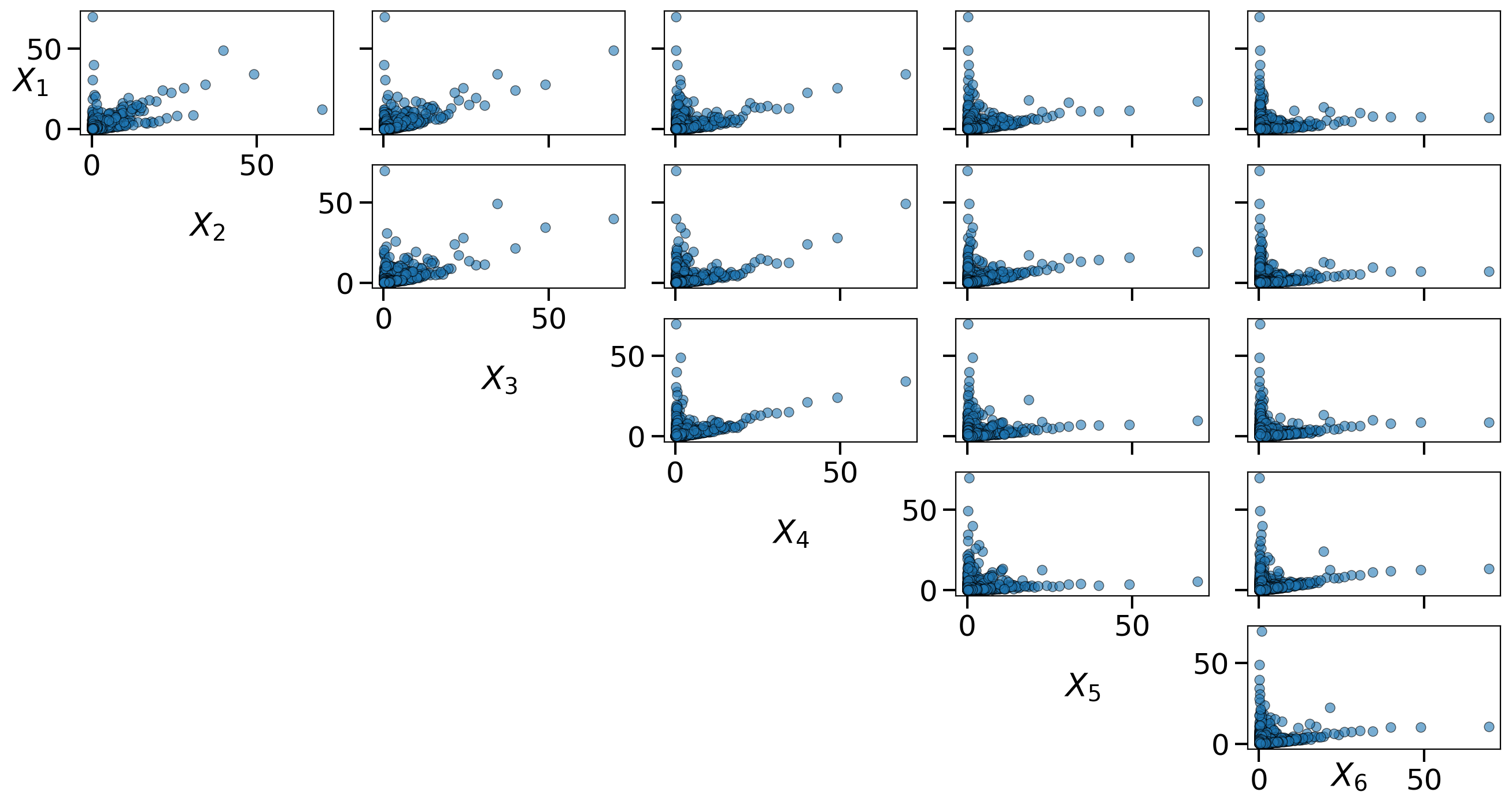}
    \caption{Pairwise scatter plots of the six components of $(X_1, \dots, X_6)^\top$ generated from the example extremal Markov network in Figure~\ref{fig:extremalmarkov:illustration}. Each variable is marginally transformed to a Pareto distribution with unit scale and shape parameter $2$.}
    \label{fig:extremalmarkov:scatter}
    \end{figure}

\subsection{Estimating DAGs for XSCMs from time series data}
\label{sec:appendix:simulation:DAGtime}

Figure~\ref{fig:XSCMts:illustration} shows an example of randomly initialized $\bm{B}_{(0)}$ and $\bm{B}_{(1)}$ for a 3-dimensional time series with lag $\tau = 1$, and the corresponding DAG.
Figure~\ref{fig:XSCMts:scatter} displays the simulated trajectory (length $T=5000$) corresponding to the XSCM in Figure~\ref{fig:XSCMts:illustration}, where occurrences of extreme values can be observed across multiple variables at the same time.

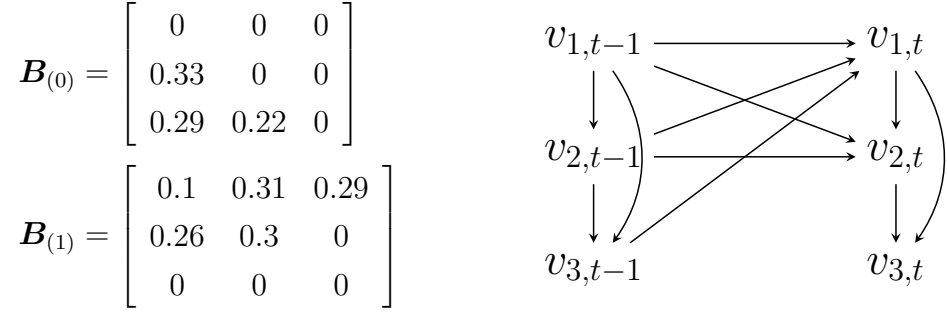
\begin{figure}[H]
    \centering
    % 第一个子图和第一个矩阵
    \begin{minipage}{0.48\textwidth}
        \vspace{-0.5cm}
        \centering
\[
\renewcommand{\arraystretch}{0.7} 
\begin{aligned}
\bm{B}_{(0)} &=
\left[
\begin{array}{ccc}
0    & 0    & 0    \\
0.33 & 0    & 0    \\
0.29 & 0.22 & 0    
\end{array}
\right] \\[0.2em]
\bm{B}_{(1)} &=
\left[
\begin{array}{ccc}
0.1  & 0.31 & 0.29 \\
0.26 & 0.3  & 0    \\
0    & 0    & 0    
\end{array}
\right]
\end{aligned}
\]
    \end{minipage}
    \hspace{-1cm}
    \begin{minipage}{0.48\textwidth}
        \centering
        \begin{tikzpicture}

            \node (v1t1) at (0,3)  {\Large $v_{1,t-1}$};
            \node (v2t1) at (0,1.5)  {\Large $v_{2,t-1}$};
            \node (v3t1) at (0,0)    {\Large $v_{3,t-1}$};        
            \node (v1t0) at (4,3)  {\Large $v_{1,t}$};
            \node (v2t0) at (4,1.5)  {\Large $v_{2,t}$};
            \node (v3t0) at (4,0)    {\Large $v_{3,t}$};
            \draw[->, line width=0.2mm, >=stealth] (v1t1) -- (v2t1); 
            \draw[->, line width=0.2mm, >=stealth] (v2t1) -- (v3t1);
            \draw[->, line width=0.2mm, >=stealth, bend left=35] (v1t1) to (v3t1);
            \draw[->, line width=0.2mm, >=stealth] (v1t0) -- (v2t0); 
            \draw[->, line width=0.2mm, >=stealth] (v2t0) -- (v3t0);
            \draw[->, line width=0.2mm, >=stealth, bend left=35] (v1t0) to (v3t0);
            \draw[->, line width=0.2mm, >=stealth] (v1t1) -- (v1t0);
            \draw[->, line width=0.2mm, >=stealth] (v2t1) -- (v1t0);
            \draw[->, line width=0.2mm, >=stealth] (v3t1) -- (v1t0);
            \draw[->, line width=0.2mm, >=stealth] (v2t1) -- (v2t0);
            \draw[->, line width=0.2mm, >=stealth] (v1t1) -- (v2t0);  
        \end{tikzpicture}
    \end{minipage} 
\caption{An example of randomly initialized path coefficient matrices $\bm{B}_{(0)}$ and $\bm{B}_{(1)}$ for an XSCM with three time series variables (left), and the corresponding DAG structure (right).}
\label{fig:XSCMts:illustration}
\end{figure}

\begin{figure}[H]
    \centering
    \includegraphics[width=0.95\textwidth]{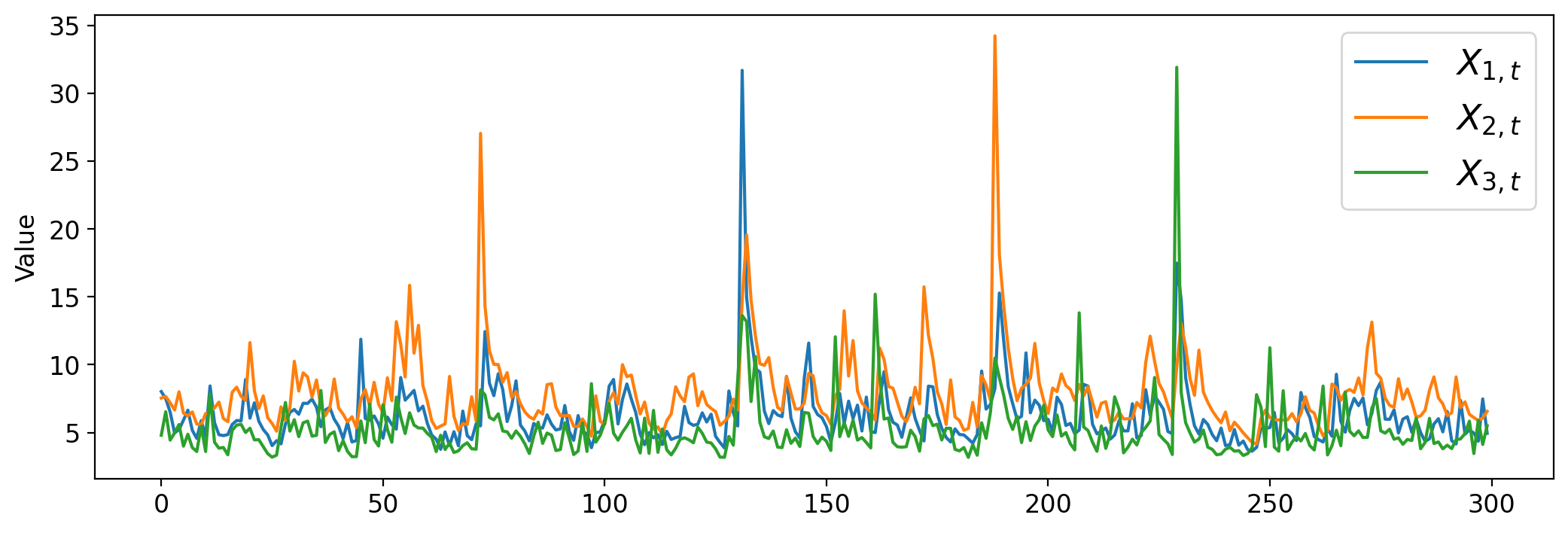}
    \caption{Simulated time series trajectory from the XSCM in Figure~\ref{fig:XSCMts:illustration}. Each time series $X_{i,t}$ is shown in different colours.}
    \label{fig:XSCMts:scatter}
\end{figure}

\subsection{DAG estimation under an alternative data-generating mechanism: max-linear model}
\label{sec:sim:XSCM:misspecific}

Max-linear structural causal models \citep{krali2025causal} assume structural equations of the form
\begin{equation}    
X_i = \bigvee_{X_j \in \mathrm{Pa}(X_i)}  \beta_{j\rightarrow i} X_j  \vee Z_i ,
\end{equation}
where $\mathrm{Pa}(X_i)$ denotes the set of causal parents of variable $X_i$, $\beta_{j\rightarrow i} > 0$ are edge weights, and $Z_i$ are independent noise variables with heavy-tailed distributions.

To evaluate the performance of our separation-based approach under model misspecification, we conduct the same experiment as in Section~\ref{sec:sim:XSCM:non-time}, but generate the data from the max-linear model instead of the XSCM. The source variables and the randomly initialized path coefficient matrix $\bm{B}$ follow the same settings as in the main text. The results are shown in Figure~\ref{fig:DAG_evaluation:misspecification}.
\begin{figure}[H]
    \centering
    \begin{subfigure}{0.475\textwidth}
        \centering
        \includegraphics[width=\textwidth]{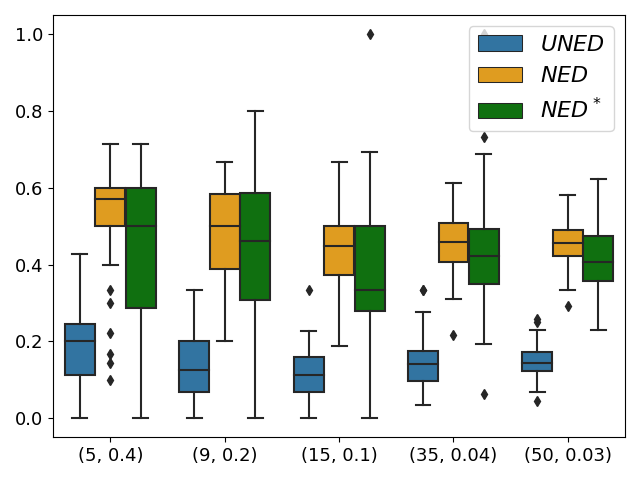}
    \end{subfigure}
    \hspace{0.5cm}
    % 右边的子图
    \begin{subfigure}{0.475\textwidth}
        \centering
        \includegraphics[width=\textwidth]{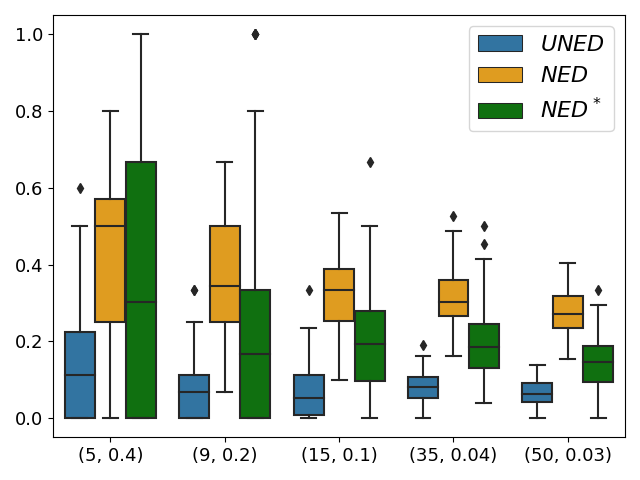} 
    \end{subfigure}
    \caption{Box plots of error measured by NED, UNED, and NED$^*$ for sample size $n=5000$ (left) and $n=50000$ (right)  under varying scale and connectivity parameters $(p, \phi)$ when data are generated from the max-linear model. Each setting is repeated for $M = 50$ randomly initialized DAGs. }
    \label{fig:DAG_evaluation:misspecification}
\end{figure}

The performance is similar to the results obtained in Section~\ref{sec:appendix:simulation:DAGnontime}. This indicates that the proposed separation-based approach is robust and can still recover the DAG when the true data-generating process deviates from the assumed XSCM model.

\subsection{Undirected graph estimation under an alternative data-generating mechanism: H\"usler--Reiss model}
\label{sec:sim:extremalNetwork:misspecification}

Another data-generating mechanism for extremal graphical models that we consider is the multivariate H\"usler--Reiss model \citep{engelke2021learning,engelke2020graphical}. To evaluate the robustness of our method for estimating the undirected graph structure, we randomly generate $M=50$ H\"usler--Reiss models and simulate samples of size $n=5000$. \citet{engelke2021learning} developed the R package \texttt{graphicalExtremes}, which provides convenient tools for randomly generating graphical H\"usler--Reiss models and simulating data from them. For example, a randomly initialized model and simulated data can be generated using the following code:

\begin{lstlisting}[language=R]
graph <- generate_random_connected_graph(num_nodes, p = edge_probability)
gamma <- generate_random_graphical_Gamma(graph)
simulate_data <- rmpareto(n = n, model = "HR", par = gamma)
\end{lstlisting}

We then apply our method to estimate the graph structure and compare the results with the method \texttt{eglearn} introduced in \citet{engelke2021learning}. The comparison is repeated for different numbers of nodes and sparsity levels of the graph. We also conduct experiments with a larger sample size $n=50000$.

The comparison results are visualized in Figure~\ref{fig:emn_evaluation:misspecification}. When the sample size is $n=5000$, the performance of \texttt{eglearn} is better than that of our method. This is expected since \texttt{eglearn} is specifically designed for learning graphical structures under the H\"usler--Reiss model. However, as the sample size increases to $n=50000$, the performance gap between the two methods decreases substantially. This suggests that although our method is not tailored to the H\"usler--Reiss model, it can still recover the underlying extremal graphical structure when sufficient data are available, demonstrating robustness to model misspecification.

\begin{figure}[t]
    \centering
    \begin{subfigure}{0.475\textwidth}
        \centering
        \includegraphics[width=\textwidth]{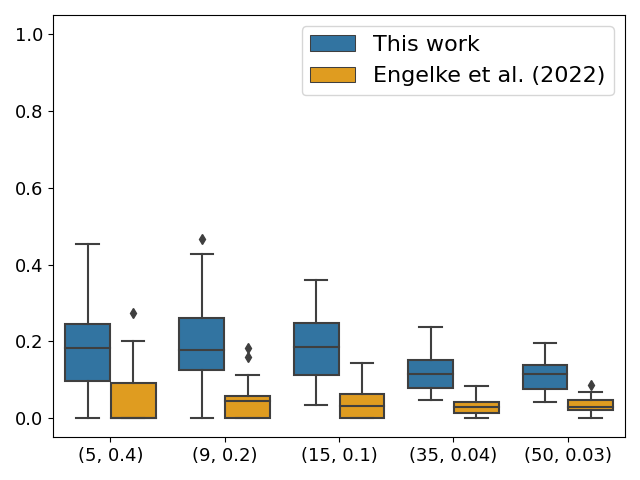}
    \end{subfigure}
    \hspace{0.5cm}
    % 右边的子图
    \begin{subfigure}{0.475\textwidth}
        \centering
        \includegraphics[width=\textwidth]{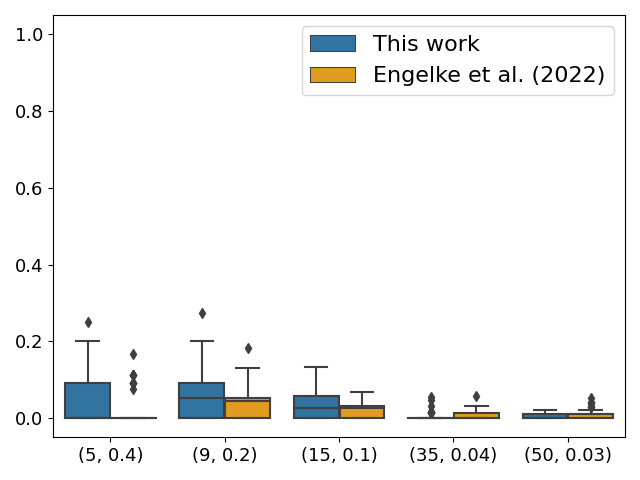} 
    \end{subfigure}
        \caption{Box plots of $\mbox{UNED}(\mathcal{E}, \hat{\mathcal{E}})$ for sample sizes $n=5000$ (left) and $n=50000$ (right) under different combinations of scale and connectivity parameters $(p,\phi)$. Data are generated from the multivariate H\"usler--Reiss model. The results are compared with the method \texttt{eglearn} proposed by \citet{engelke2021learning}. Each setting is repeated for $M=50$ randomly generated graphs. The parameters for \citet{engelke2021learning} are set to their default values, while our method uses the fixed parameters described in Section~\ref{sec:sim:overview}.}
        \label{fig:emn_evaluation:misspecification}
\end{figure}

\section{Proofs}
\label{sec:Proof}

%\subsection{Proof of Proposition~\ref{prop:tpdmForOnetails}}\label{sec:Proof}

\begin{proof}[Proof of Proposition~\ref{prop:tpdmForOnetails}]

We first prove that $(\bm{I}-\bm{B})$ is invertible and that its inverse, $(\bm{I}-\bm{B})^{-1},$ is a non-negative matrix. Assuming this holds, the direct calculation 
\begin{equation*}
    \begin{aligned}
        \bm{X}&=\bm{A}\circ \bm{Z} \oplus \bm{B}\circ \bm{X},\\
    t^{-1}(\bm{X})&=\bm{A}t^{-1}(\bm{Z})+\bm{B}t^{-1}(\bm{X}),\\
    (\bm{I}-\bm{B})t^{-1}(\bm{X})&=\bm{A}t^{-1}(\bm{Z}),\\
    \bm{X}&=(\bm{I}-\bm{B})^{-1}\bm{A}\circ \bm{Z}\\
    \end{aligned}
    \end{equation*}
shows that $\bm{X}=(\bm{I}-\bm{B})^{-1}\bm{A}\circ \bm{Z}$. For any matrix $\bm{D}=(D_{ij})_{i,j=1}^p\in \mathbb{R}^{p\times p}$, the TPDM for $\bm{D}\circ\bm{Z}$ is $\bm{D}^{(0)}{(\bm{D}^{(0)})}^\top$ where $\bm{D}^{(0)}:=\{\max(D_{ij},0)\}_{i,j=1,\dots,p}$ masks the negative values in $\bm{D}$~\citep[][Sec.~4]{cooley2019decompositions}. Courtesy of the non-negativity of $(\bm{I}-\bm{B})^{-1}$, we obtain $$\bm{\Sigma}_{\bm{X}}=\left((\bm{I}-\bm{B})^{-1}\bm{A}\right)^{(0)}\left(\left((\bm{I}-\bm{B})^{-1}\bm{A}\right)^{(0)}\right)^\top=(\bm{I}-\bm{B})^{-1}\bm{A}^{2}[(\bm{I}-\bm{B})^{-1}]^\top,$$ where the $\bm{A}$ is the scale matrix of the XSCM.

Now, to prove the invertibility and non-negativity of $(\bm{I}-\bm{B})^{-1}$, we first reorder the random vector $\bm{X}$ and $\bm{Z}$ using \textit{topological ordering}, which is a bijective mapping $\pi:\{1,\dots,p\}\mapsto \{1,\dots,p\}$ that satisfies $\pi(j)<\pi(i)$ whenever $v_j\rightarrow v_i$. For any DAG $\mathcal{G}:=(\mathcal{V},\mathcal{E})$, a topological ordering always exists (\citealp{peters2017elements};~Proposition~{B.2} and \citealp{Cormen_Leiserson_2022};~Theorem~20.12). Let $\bm{E}=(E_{ij})_{i,j=1}^p\in \{0,1\}^{p\times p}$ be the row swapping matrix for topological ordering, i.e., $E_{\pi(i),i}=1$ for $i=1,\dots,p$ and all other elements are zero. Define the reordered random vector $\bm{X}^\prime:=\bm{E}\circ\bm{X}=\bm{E}\bm{X}$, and its corresponding reordered source variables by $\bm{Z}^\prime:=\bm{E}\circ \bm{Z}=\bm{E}\bm{Z}$. Rewriting $\bm{X}^\prime$ and $\bm{Z}^\prime$ using $\bm{B}$ and $\bm{A}$, we have that
$$\bm{X}^\prime=\bm{E}\bm{B}\bm{E}^{-1}\circ \bm{X}^\prime\oplus \bm{E}\bm{A}\bm{E}^{-1}\circ \bm{Z}^\prime.$$

The corresponding reordered path coefficient matrix is $\bm{B}^\prime=\bm{E}\bm{B}\bm{E}^{-1}$ and the scale matrix is $\bm{A}^\prime=\bm{E}\bm{A}\bm{E}^{-1}$ for $\bm{X}^\prime$ and $\bm{Z}^\prime$. By the definition of topological ordering, $\bm{B}^\prime$ is a strict lower triangular matrix. Consequently, $\bm{B}^\prime$ is a convergent matrix, i.e., ${\lim_{k\to \infty}(\bm{B}^\prime)^k=\bm{0}}$, and has a spectral radius ${\rho(\bm{B}^\prime)<1}$~\citep[][Theorem~4.c]{Isaacson_1994}. Since $\bm{B}^\prime$ and $\bm{B}$ are similar matrices, the path coefficient matrix $\bm{B}$ is also convergent.

For a convergent matrix $\bm{B}$, the matrix $(\bm{I}-\bm{B})$ is invertible and can be expressed as a Neumann series: ${(\bm{I}-\bm{B})^{-1}=\sum_{k=0}^{\infty}\bm{B}^k}$~\citep[][Theorem~5]{Isaacson_1994}. Since $\bm{B}$ is non-negative, $\bm{B}^k$ remains non-negative for all $k\geq 0$, implying that $(\bm{I}-\bm{B})^{-1}$ is also non-negative.
\end{proof}

\begin{remark}
    An alternative proof is to show that $(\bm{I}-\bm{B})$ is a nonsingular $\mathcal{M}$-type matrix~\citep[][Definition~1.2]{Berman_1994}. A nonsingular $\mathcal{M}$-type matrix is a square matrix in the form $s\bm{I}-\bm{D}$ where $\rho(\bm{D})<s$ for $s>0,\bm{D}\in \mathbb{R}_+^{p\times p}$. An $\mathcal{M}$-type matrix is invertible and its inverse is positive. After deriving $\rho(\bm{B})<1$, we readily notice that $(\bm{I}-\bm{B})$ is a nonsingular $\mathcal{M}$-type matrix.
\end{remark}
% Page 133 Definition~1.2{Berman_1994}

%\subsection{Proof of Proposition~\ref{prop:Mar}}\label{proof:Mar}

To prove Proposition~\ref{prop:Mar}, we need the following lemma. 

\begin{lemma}\label{lem:XSCM-linearCausalModel}
    For any XSCM $\mathcal{S}:=(\bm{Z},\mathcal{F})$ for $\bm{X}=(X_1,\dots,X_p)^\top$ with underlying graph $\mathcal{G}=(\mathcal{V},\mathcal{E})$ and ${\bm{X}=(\bm{I}-\bm{B})^{-1}\bm{A}\circ \bm{Z}}$, there exists a Gaussian linear causal model $\mathcal{S}^\prime:=(\bm{Z}^\prime,\mathcal{F}^\prime_1)$ for $\bm{X}^\prime=(X_1^\prime,\dots,X_p^\prime)^\top$, where the source variables $\bm{Z}^\prime=(Z_1^\prime,\dots,Z_p^\prime)^\top$ are independent and identically distributed Gaussian random variables. The model $\mathcal{S}^\prime$ shares the same graph structure $\mathcal{G}=(\mathcal{V},\mathcal{E})$ and the same path coefficient matrix $\bm{B}$ and scale matrix $\bm{A}$ as $\mathcal{S}$. Moreover, the direct representation of $\bm{X}^\prime$ is $\bm{X}^\prime=(\bm{I}-\bm{B})^{-1}\bm{A}\bm{Z}^\prime$.
\end{lemma}

\begin{proof}[Proof of Lemma~\ref{lem:XSCM-linearCausalModel}]
    This follows from the existence of an equivalent Gaussian linear structural causal model that shares the same path coefficient matrix $\bm{B}$ and scale matrix $\bm{A}$ as the given XSCM, which follows directly from the standard definition of Gaussian linear structural causal models (see, e.g., \cite{spirtes2000causation}) together with our general definition of $\bm{B}$. In particular, $\bm{X}^\prime$ can be written as $\bm{X}^\prime = (\bm{I}-\bm{B})^{-1}\bm{A}\bm{Z}^\prime$, and the invertibility of $\bm{I}-\bm{B}$ is guaranteed by Proposition~\ref{prop:tpdmForOnetails}.
\end{proof}

\begin{proof}[Proof of Proposition~\ref{prop:Mar}]

For any XSCM $\mathcal{S}:=(\bm{Z},\mathcal{F})$ for $\bm{X}=(X_1,\dots,X_p)^\top$, where $\bm{X}=(\bm{I}-\bm{B})^{-1}\bm{A}\circ \bm{Z}$, let $\mathcal{S}^\prime:=(\bm{Z}^\prime,\mathcal{F}^\prime)$ for $\bm{X}^\prime=(X_1^\prime,\dots,X_p^\prime)^\top$ be the corresponding Gaussian linear causal model as stated in Lemma~\ref{lem:XSCM-linearCausalModel}.
They share the same underlying graph structure $\mathcal{G}=(\mathcal{V},\mathcal{E})$ and the same path coefficient matrix $\bm{B}$ and scale matrix $\bm{A}$ as $\mathcal{S}$. The system $\bm{X}^\prime=(\bm{I}-\bm{B})^{-1}\bm{A}\bm{Z}^\prime$ is a linear function of $\bm{Z}^\prime$ and follows a joint Gaussian distribution.

It is well established that Gaussian linear structural causal models satisfy both the causal Markov and causal faithfulness conditions~\citep[][Theorem~3.2 and Proposition~6.31]{spirtes2000causation,peters2017elements}. Consequently, for any two vertices $v_i,v_j\in \mathcal{V}$ and $\bm{\mathcal{S}}_v\subset \mathcal{V}\setminus \{v_i,v_j\}$, vertices $v_i$ and $v_j$ are d-separated by $\bm{\mathcal{S}}_v$ if and only if the conditional covariance $ {\rm Cov}[X^\prime_i,X^\prime_j\mid\bm{\mathcal{S}}^\prime]=0$, where $\bm{\mathcal{S}}^\prime=\{X_k^\prime\}_{v_k\in \bm{\mathcal{S}}_v}$.

Furthermore, the conditional covariance ${\rm Cov}[X_i,X_j\mid\bm{\mathcal{S}}^\prime]$ has the same form as the partial tail-covariance $\sigma_{ij|\bm{\mathcal{S}}}$ where $\bm{\mathcal{S}}=\{X_k\}_{v_k\in \bm{\mathcal{S}}_v}$. Therefore, for any $v_i,v_j\in \mathcal{V}$ and $\bm{\mathcal{S}}_v\subset \mathcal{V}\setminus \{v_i,v_j\}$, $v_i \perp_{\mathcal{G}} v_j\mid\bm{\mathcal{S}}_v$ if and only if $\sigma_{ij| \bm{\mathcal{S}}}=0$ where $\bm{\mathcal{S}}=\{X_k\}_{v_k\in \bm{\mathcal{S}}_v}$. This equivalence establishes both the tail-causal Markov and tail-causal faithfulness properties for XSCM.
\end{proof}

%\subsection{Proof of Proposition~\ref{prop:extrGlobalMarkov}}\label{proof:extrGlobalMarkov}

\begin{proof}[Proof of Proposition~\ref{prop:extrGlobalMarkov}]
    Given a random vector $\bm{X}=(X_1,\dots,X_p)^\top$ with ${X_i=\bm{a_i}\circ \bm{Z}\in V_+^q}$ where $\bm{Z}=(Z_1,\dots,Z_q)$ are independent and identically-distributed regularly varying random variables with tail index $\alpha=2$. Let ${\bm{A}=(\bm{a}_1,\dots,\bm{a}_p)^\top}$. The TPDM of $\bm{X}$ is given by $\bm{\Sigma}_{\bm{X}}=\bm{A}\bm{A}^\top$ and the induced undirected graph is $\check{\mathcal{G}}=(\mathcal{V},\mathcal{E})$.

    Let a corresponding random vector $\bm{X}^\prime=\bm{A}\bm{Z}^\prime$ with $\bm{Z}^\prime\sim \mathcal{N}_p(0,\bm{I})$. The vector $\bm{X}^\prime$ is a Gaussian Markov random field~\citep[][Definition~2.1]{rue2005gaussian} with regards to the same underlying undirected graph $\check{\mathcal{G}}=(\mathcal{V},\mathcal{E})$ for the extremal Markov network $\bm{X}$. Furthermore, the covariance matrix for $\bm{X}^\prime$ is given by ${{\rm Cov}[\bm{X}^\prime, \bm{X}^\prime]=\bm{A}\bm{A}^\top=\bm{\Sigma}_{\bm{X}}}$. For Gaussian Markov random fields, the equivalence between the global Markov property and the pairwise Markov property are established~\citep[][Theorem~2.4]{rue2005gaussian}. Also, for any two vertces $v_i,v_j\in \mathcal{V}$ and a separation set $\bm{\mathcal{S}}_v\subset \mathcal{V}\setminus \{v_i,v_j\}$, the conditional covariance ${\rm Cov}[X_i^\prime,X_j^\prime\mid\bm{\mathcal{S}}^\prime]$ equals the partial tail-covariance $\sigma_{ij|\bm{\mathcal{S}}}$, where $\bm{\mathcal{S}}^\prime=\{X_k^\prime\}_{v_k}\in \bm{\mathcal{S}}_v$ and $\bm{\mathcal{S}}=\{X_k\}_{v_k\in \bm{\mathcal{S}}_v}$. For a joint Gaussian random vector, the conditional covariance equals zero if and only if conditional independence holds. Therefore, the two vertices $v_i,v_j$ are separated by $\bm{\mathcal{S}}_v$ if and only if PTCC $\gamma_{ij|\bm{\mathcal{S}}}=0$ for $\bm{\mathcal{S}}=\{X_k\}_{v_k\in \bm{\mathcal{S}}_v}$. This concludes the proof.
\end{proof}
    
\section{Supplementary numerical experiments}
\label{sec:supplementaryNumericalExperiments}

\subsection{Sensitivity analysis of $\alpha$ and $q$ in Section~\ref{sec:sim:XSCM:non-time}}\label{sec:sensitivity}
To investigate the sensitivity of the DAG estimation to the choice of hyperparameters $\alpha$ and $q$, we consider the scenario of estimating the DAG of an XSCM model using cross-sectional data. We conduct experiments under different settings of the number of variables $p$ and sparsity level $\phi$, and report the recovery performance—measured by UNED, NED, and NED$^*$---across various choices of $\alpha$ and $q$. Figure~\ref{fig:sensitivity} presents the results. The performance is not sensitive to the significance level $\alpha$, and it remains stable with respect to the quantile $q$ when $q$ is close to 1.
\begin{figure}[ht!]
    \centering
    \includegraphics[width=\textwidth]{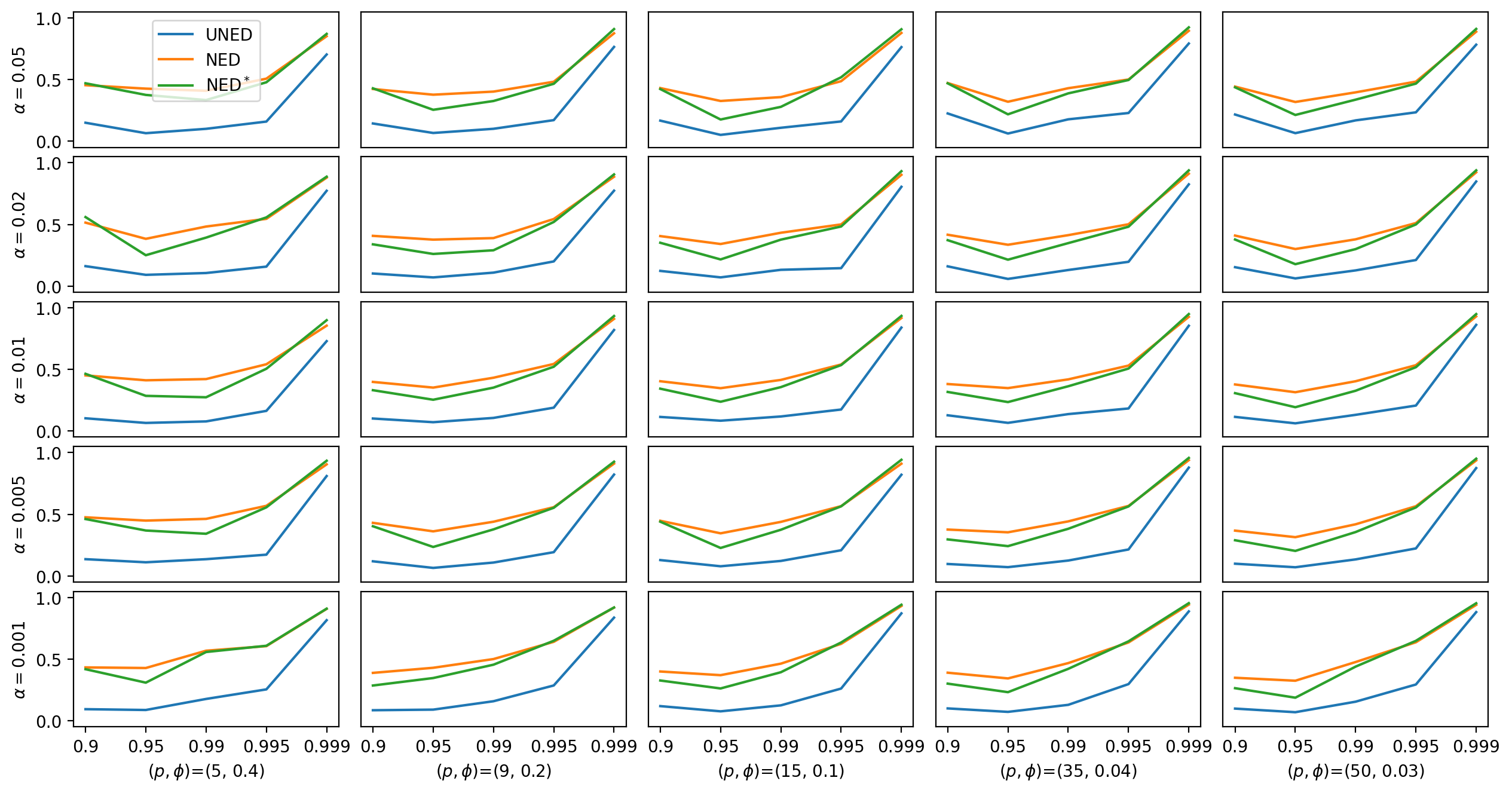}
    \caption{Sensitivity analysis of the hyperparameters $\alpha$ (rows) and $q$ (x-axis) for the DAG estimation described in Section~\ref{sec:sim:XSCM:non-time} under different settings of variable $p$ and sparsity level $\phi$ (columns); performance is measured by UNED, NED, and NED$^*$.}
    \label{fig:sensitivity}
\end{figure}

\subsection{Sensitivity analysis of parameters for the Danube river data}\label{sec:danubeSensitivity}

To investigate the impact of tuning parameters on the estimated causal structure, we conduct a sensitivity analysis on the Danube river dataset by varying the quantile level $q$, the significance level $\alpha$, and the lag parameter $\tau$. Figures~\ref{fig:danubeSensitivityTau0}--\ref{fig:danubeSensitivityTau2} report the estimated graphs under different combinations of these parameters.

A clear pattern emerges across all values of $\tau$: the estimated graphs become increasingly sparse as the significance level $\alpha$ decreases and as the quantile level $q$ increases. For the temporal case $\tau = 1,2$, the inferred lagged causal directions predominantly point from upstream to downstream stations, which is consistent with the physical flow direction of the river system.

Overall, the estimated graph obtained at $q = 0.99$ and $\alpha = 0.005$ achieves a favorable balance between sparsity and interpretability. This parameter configuration is adopted in the Danube river data application.

% Sensitivity over quantile $q$ (rows) and $\alpha$ (columns) for Danube river, $\tau = 0, 1, 2$.
\begin{figure}[htbp]
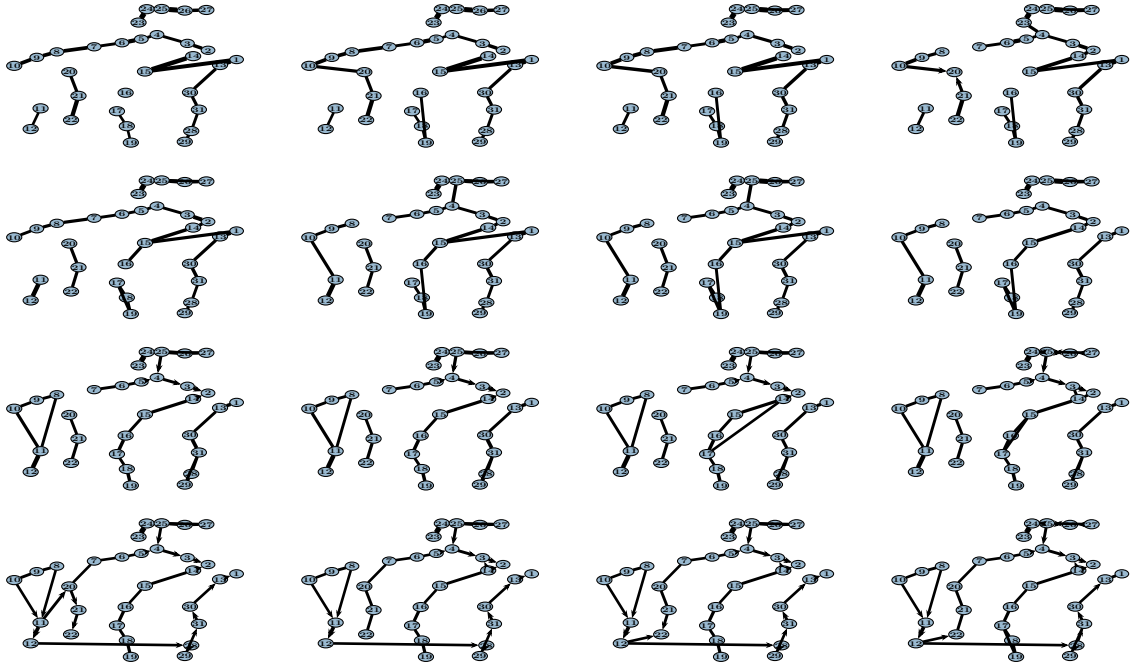

    \centering
    \begin{subfigure}[b]{\textwidth}
        \centering
        \begin{tikzpicture}[scale=0.1471,transform shape,every node/.style={circle, draw, outer sep=0pt,minimum size=35pt,fill=danubeBlue!50, inner sep=0pt,font=\bfseries\fontsize{30}{35}\selectfont}]
            \begin{scope}[yscale=0.601]
            \input{pic/370.Danbube_river/ours0_0.5_0.001.tex}
            \end{scope}
        \end{tikzpicture}\hspace{1.5em}
        \begin{tikzpicture}[scale=0.1471,transform shape,every node/.style={circle, draw, outer sep=0pt,minimum size=35pt,fill=danubeBlue!50, inner sep=0pt,font=\bfseries\fontsize{30}{35}\selectfont}]
            \begin{scope}[yscale=0.601]
            \input{pic/370.Danbube_river/ours0_0.5_0.005.tex}
            \end{scope}
        \end{tikzpicture}\hspace{1.5em}
        \begin{tikzpicture}[scale=0.1471,transform shape,every node/.style={circle, draw, outer sep=0pt,minimum size=35pt,fill=danubeBlue!50, inner sep=0pt,font=\bfseries\fontsize{30}{35}\selectfont}]
            \begin{scope}[yscale=0.601]
            \input{pic/370.Danbube_river/ours0_0.5_0.01.tex}
            \end{scope}
        \end{tikzpicture}\hspace{1.5em}
        \begin{tikzpicture}[scale=0.1471,transform shape,every node/.style={circle, draw, outer sep=0pt,minimum size=35pt,fill=danubeBlue!50, inner sep=0pt,font=\bfseries\fontsize{30}{35}\selectfont}]
            \begin{scope}[yscale=0.601]
            \input{pic/370.Danbube_river/ours0_0.5_0.05.tex}
            \end{scope}
        \end{tikzpicture}
    \end{subfigure}\\[0.8em]
    \begin{subfigure}[b]{\textwidth}
        \centering
        \begin{tikzpicture}[scale=0.1471,transform shape,every node/.style={circle, draw, outer sep=0pt,minimum size=35pt,fill=danubeBlue!50, inner sep=0pt,font=\bfseries\fontsize{30}{35}\selectfont}]
            \begin{scope}[yscale=0.601]
            \input{pic/370.Danbube_river/ours0_1_0.001.tex}
            \end{scope}
        \end{tikzpicture}\hspace{1.5em}
        \begin{tikzpicture}[scale=0.1471,transform shape,every node/.style={circle, draw, outer sep=0pt,minimum size=35pt,fill=danubeBlue!50, inner sep=0pt,font=\bfseries\fontsize{30}{35}\selectfont}]
            \begin{scope}[yscale=0.601]
            \input{pic/370.Danbube_river/ours0_1_0.005.tex}
            \end{scope}
        \end{tikzpicture}\hspace{1.5em}
        \begin{tikzpicture}[scale=0.1471,transform shape,every node/.style={circle, draw, outer sep=0pt,minimum size=35pt,fill=danubeBlue!50, inner sep=0pt,font=\bfseries\fontsize{30}{35}\selectfont}]
            \begin{scope}[yscale=0.601]
            \input{pic/370.Danbube_river/ours0_1_0.01.tex}
            \end{scope}
        \end{tikzpicture}\hspace{1.5em}
        \begin{tikzpicture}[scale=0.1471,transform shape,every node/.style={circle, draw, outer sep=0pt,minimum size=35pt,fill=danubeBlue!50, inner sep=0pt,font=\bfseries\fontsize{30}{35}\selectfont}]
            \begin{scope}[yscale=0.601]
            \input{pic/370.Danbube_river/ours0_1_0.05.tex}
            \end{scope}
        \end{tikzpicture}
    \end{subfigure}\\[0.8em]
    \begin{subfigure}[b]{\textwidth}
        \centering
        \begin{tikzpicture}[scale=0.1471,transform shape,every node/.style={circle, draw, outer sep=0pt,minimum size=35pt,fill=danubeBlue!50, inner sep=0pt,font=\bfseries\fontsize{30}{35}\selectfont}]
            \begin{scope}[yscale=0.601]
            \input{pic/370.Danbube_river/ours0_5_0.001.tex}
            \end{scope}
        \end{tikzpicture}\hspace{1.5em}
        \begin{tikzpicture}[scale=0.1471,transform shape,every node/.style={circle, draw, outer sep=0pt,minimum size=35pt,fill=danubeBlue!50, inner sep=0pt,font=\bfseries\fontsize{30}{35}\selectfont}]
            \begin{scope}[yscale=0.601]
            \input{pic/370.Danbube_river/ours0_5_0.005.tex}
            \end{scope}
        \end{tikzpicture}\hspace{1.5em}
        \begin{tikzpicture}[scale=0.1471,transform shape,every node/.style={circle, draw, outer sep=0pt,minimum size=35pt,fill=danubeBlue!50, inner sep=0pt,font=\bfseries\fontsize{30}{35}\selectfont}]
            \begin{scope}[yscale=0.601]
            \input{pic/370.Danbube_river/ours0_5_0.01.tex}
            \end{scope}
        \end{tikzpicture}\hspace{1.5em}
        \begin{tikzpicture}[scale=0.1471,transform shape,every node/.style={circle, draw, outer sep=0pt,minimum size=35pt,fill=danubeBlue!50, inner sep=0pt,font=\bfseries\fontsize{30}{35}\selectfont}]
            \begin{scope}[yscale=0.601]
            \input{pic/370.Danbube_river/ours0_5_0.05.tex}
            \end{scope}
        \end{tikzpicture}
    \end{subfigure}\\[0.8em]
    \begin{subfigure}[b]{\textwidth}
        \centering
        \begin{tikzpicture}[scale=0.1471,transform shape,every node/.style={circle, draw, outer sep=0pt,minimum size=35pt,fill=danubeBlue!50, inner sep=0pt,font=\bfseries\fontsize{30}{35}\selectfont}]
            \begin{scope}[yscale=0.601]
            \input{pic/370.Danbube_river/ours0_10_0.001.tex}
            \end{scope}
        \end{tikzpicture}\hspace{1.5em}
        \begin{tikzpicture}[scale=0.1471,transform shape,every node/.style={circle, draw, outer sep=0pt,minimum size=35pt,fill=danubeBlue!50, inner sep=0pt,font=\bfseries\fontsize{30}{35}\selectfont}]
            \begin{scope}[yscale=0.601]
            \input{pic/370.Danbube_river/ours0_10_0.005.tex}
            \end{scope}
        \end{tikzpicture}\hspace{1.5em}
        \begin{tikzpicture}[scale=0.1471,transform shape,every node/.style={circle, draw, outer sep=0pt,minimum size=35pt,fill=danubeBlue!50, inner sep=0pt,font=\bfseries\fontsize{30}{35}\selectfont}]
            \begin{scope}[yscale=0.601]
            \input{pic/370.Danbube_river/ours0_10_0.01.tex}
            \end{scope}
        \end{tikzpicture}\hspace{1.5em}
        \begin{tikzpicture}[scale=0.1471,transform shape,every node/.style={circle, draw, outer sep=0pt,minimum size=35pt,fill=danubeBlue!50, inner sep=0pt,font=\bfseries\fontsize{30}{35}\selectfont}]
            \begin{scope}[yscale=0.601]
            \input{pic/370.Danbube_river/ours0_10_0.05.tex}
            \end{scope}
        \end{tikzpicture}
    \end{subfigure}
    \caption{Estimated graphs for Danube river dataset with $\tau = 0$ across different quantile level $q$ (rows: $0.995$, $0.99$, $0.95$, $0.9$) and $\alpha$ (columns: $0.001$, $0.005$, $0.01$, $0.05$).}
    \label{fig:danubeSensitivityTau0}
\end{figure}

\begin{figure}[htbp]
    \centering
    \begin{subfigure}[b]{\textwidth}
        \centering
        \begin{tikzpicture}[scale=0.1471,transform shape,every node/.style={circle, draw, outer sep=0pt,minimum size=35pt,fill=danubeBlue!50, inner sep=0pt,font=\bfseries\fontsize{30}{35}\selectfont}]
            \begin{scope}[yscale=0.601]
            \input{pic/370.Danbube_river/ours1_0.5_0.001.tex}
            \end{scope}
        \end{tikzpicture}\hspace{1.5em}
        \begin{tikzpicture}[scale=0.1471,transform shape,every node/.style={circle, draw, outer sep=0pt,minimum size=35pt,fill=danubeBlue!50, inner sep=0pt,font=\bfseries\fontsize{30}{35}\selectfont}]
            \begin{scope}[yscale=0.601]
            \input{pic/370.Danbube_river/ours1_0.5_0.005.tex}
            \end{scope}
        \end{tikzpicture}\hspace{1.5em}
        \begin{tikzpicture}[scale=0.1471,transform shape,every node/.style={circle, draw, outer sep=0pt,minimum size=35pt,fill=danubeBlue!50, inner sep=0pt,font=\bfseries\fontsize{30}{35}\selectfont}]
            \begin{scope}[yscale=0.601]
            \input{pic/370.Danbube_river/ours1_0.5_0.01.tex}
            \end{scope}
        \end{tikzpicture}\hspace{1.5em}
        \begin{tikzpicture}[scale=0.1471,transform shape,every node/.style={circle, draw, outer sep=0pt,minimum size=35pt,fill=danubeBlue!50, inner sep=0pt,font=\bfseries\fontsize{30}{35}\selectfont}]
            \begin{scope}[yscale=0.601]
            \input{pic/370.Danbube_river/ours1_0.5_0.05.tex}
            \end{scope}
        \end{tikzpicture}
    \end{subfigure}\\[0.8em]
    \begin{subfigure}[b]{\textwidth}
        \centering
        \begin{tikzpicture}[scale=0.1471,transform shape,every node/.style={circle, draw, outer sep=0pt,minimum size=35pt,fill=danubeBlue!50, inner sep=0pt,font=\bfseries\fontsize{30}{35}\selectfont}]
            \begin{scope}[yscale=0.601]
            \input{pic/370.Danbube_river/ours1_1_0.001.tex}
            \end{scope}
        \end{tikzpicture}\hspace{1.5em}
        \begin{tikzpicture}[scale=0.1471,transform shape,every node/.style={circle, draw, outer sep=0pt,minimum size=35pt,fill=danubeBlue!50, inner sep=0pt,font=\bfseries\fontsize{30}{35}\selectfont}]
            \begin{scope}[yscale=0.601]
            \input{pic/370.Danbube_river/ours1_1_0.005.tex}
            \end{scope}
        \end{tikzpicture}\hspace{1.5em}
        \begin{tikzpicture}[scale=0.1471,transform shape,every node/.style={circle, draw, outer sep=0pt,minimum size=35pt,fill=danubeBlue!50, inner sep=0pt,font=\bfseries\fontsize{30}{35}\selectfont}]
            \begin{scope}[yscale=0.601]
            \input{pic/370.Danbube_river/ours1_1_0.01.tex}
            \end{scope}
        \end{tikzpicture}\hspace{1.5em}
        \begin{tikzpicture}[scale=0.1471,transform shape,every node/.style={circle, draw, outer sep=0pt,minimum size=35pt,fill=danubeBlue!50, inner sep=0pt,font=\bfseries\fontsize{30}{35}\selectfont}]
            \begin{scope}[yscale=0.601]
            \input{pic/370.Danbube_river/ours1_1_0.05.tex}
            \end{scope}
        \end{tikzpicture}
    \end{subfigure}\\[0.8em]
    \begin{subfigure}[b]{\textwidth}
        \centering
        \begin{tikzpicture}[scale=0.1471,transform shape,every node/.style={circle, draw, outer sep=0pt,minimum size=35pt,fill=danubeBlue!50, inner sep=0pt,font=\bfseries\fontsize{30}{35}\selectfont}]
            \begin{scope}[yscale=0.601]
            \input{pic/370.Danbube_river/ours1_5_0.001.tex}
            \end{scope}
        \end{tikzpicture}\hspace{1.5em}
        \begin{tikzpicture}[scale=0.1471,transform shape,every node/.style={circle, draw, outer sep=0pt,minimum size=35pt,fill=danubeBlue!50, inner sep=0pt,font=\bfseries\fontsize{30}{35}\selectfont}]
            \begin{scope}[yscale=0.601]
            \input{pic/370.Danbube_river/ours1_5_0.005.tex}
            \end{scope}
        \end{tikzpicture}\hspace{1.5em}
        \begin{tikzpicture}[scale=0.1471,transform shape,every node/.style={circle, draw, outer sep=0pt,minimum size=35pt,fill=danubeBlue!50, inner sep=0pt,font=\bfseries\fontsize{30}{35}\selectfont}]
            \begin{scope}[yscale=0.601]
            \input{pic/370.Danbube_river/ours1_5_0.01.tex}
            \end{scope}
        \end{tikzpicture}\hspace{1.5em}
        \begin{tikzpicture}[scale=0.1471,transform shape,every node/.style={circle, draw, outer sep=0pt,minimum size=35pt,fill=danubeBlue!50, inner sep=0pt,font=\bfseries\fontsize{30}{35}\selectfont}]
            \begin{scope}[yscale=0.601]
            \input{pic/370.Danbube_river/ours1_5_0.05.tex}
            \end{scope}
        \end{tikzpicture}
    \end{subfigure}\\[0.8em]
    \begin{subfigure}[b]{\textwidth}
        \centering
        \begin{tikzpicture}[scale=0.1471,transform shape,every node/.style={circle, draw, outer sep=0pt,minimum size=35pt,fill=danubeBlue!50, inner sep=0pt,font=\bfseries\fontsize{30}{35}\selectfont}]
            \begin{scope}[yscale=0.601]
            \input{pic/370.Danbube_river/ours1_10_0.001.tex}
            \end{scope}
        \end{tikzpicture}\hspace{1.5em}
        \begin{tikzpicture}[scale=0.1471,transform shape,every node/.style={circle, draw, outer sep=0pt,minimum size=35pt,fill=danubeBlue!50, inner sep=0pt,font=\bfseries\fontsize{30}{35}\selectfont}]
            \begin{scope}[yscale=0.601]
            \input{pic/370.Danbube_river/ours1_10_0.005.tex}
            \end{scope}
        \end{tikzpicture}\hspace{1.5em}
        \begin{tikzpicture}[scale=0.1471,transform shape,every node/.style={circle, draw, outer sep=0pt,minimum size=35pt,fill=danubeBlue!50, inner sep=0pt,font=\bfseries\fontsize{30}{35}\selectfont}]
            \begin{scope}[yscale=0.601]
            \input{pic/370.Danbube_river/ours1_10_0.01.tex}
            \end{scope}
        \end{tikzpicture}\hspace{1.5em}
        \begin{tikzpicture}[scale=0.1471,transform shape,every node/.style={circle, draw, outer sep=0pt,minimum size=35pt,fill=danubeBlue!50, inner sep=0pt,font=\bfseries\fontsize{30}{35}\selectfont}]
            \begin{scope}[yscale=0.601]
            \input{pic/370.Danbube_river/ours1_10_0.05.tex}
            \end{scope}
        \end{tikzpicture}
    \end{subfigure}
    \caption{Estimated graphs for Danube river dataset with $\tau = 1$ across different quantile level $q$ (rows: $0.995$, $0.99$, $0.95$, $0.9$) and $\alpha$ (columns: $0.001$, $0.005$, $0.01$, $0.05$).}
    \label{fig:danubeSensitivityTau1}
\end{figure}

\begin{figure}[htbp]
    \centering
    \begin{subfigure}[b]{\textwidth}
        \centering
        \begin{tikzpicture}[scale=0.1471,transform shape,every node/.style={circle, draw, outer sep=0pt,minimum size=35pt,fill=danubeBlue!50, inner sep=0pt,font=\bfseries\fontsize{30}{35}\selectfont}]
            \begin{scope}[yscale=0.601]
            \input{pic/370.Danbube_river/ours2_0.5_0.001.tex}
            \end{scope}
        \end{tikzpicture}\hspace{1.5em}
        \begin{tikzpicture}[scale=0.1471,transform shape,every node/.style={circle, draw, outer sep=0pt,minimum size=35pt,fill=danubeBlue!50, inner sep=0pt,font=\bfseries\fontsize{30}{35}\selectfont}]
            \begin{scope}[yscale=0.601]
            \input{pic/370.Danbube_river/ours2_0.5_0.005.tex}
            \end{scope}
        \end{tikzpicture}\hspace{1.5em}
        \begin{tikzpicture}[scale=0.1471,transform shape,every node/.style={circle, draw, outer sep=0pt,minimum size=35pt,fill=danubeBlue!50, inner sep=0pt,font=\bfseries\fontsize{30}{35}\selectfont}]
            \begin{scope}[yscale=0.601]
            \input{pic/370.Danbube_river/ours2_0.5_0.01.tex}
            \end{scope}
        \end{tikzpicture}\hspace{1.5em}
        \begin{tikzpicture}[scale=0.1471,transform shape,every node/.style={circle, draw, outer sep=0pt,minimum size=35pt,fill=danubeBlue!50, inner sep=0pt,font=\bfseries\fontsize{30}{35}\selectfont}]
            \begin{scope}[yscale=0.601]
            \input{pic/370.Danbube_river/ours2_0.5_0.05.tex}
            \end{scope}
        \end{tikzpicture}
    \end{subfigure}\\[0.8em]
    \begin{subfigure}[b]{\textwidth}
        \centering
        \begin{tikzpicture}[scale=0.1471,transform shape,every node/.style={circle, draw, outer sep=0pt,minimum size=35pt,fill=danubeBlue!50, inner sep=0pt,font=\bfseries\fontsize{30}{35}\selectfont}]
            \begin{scope}[yscale=0.601]
            \input{pic/370.Danbube_river/ours2_1_0.001.tex}
            \end{scope}
        \end{tikzpicture}\hspace{1.5em}
        \begin{tikzpicture}[scale=0.1471,transform shape,every node/.style={circle, draw, outer sep=0pt,minimum size=35pt,fill=danubeBlue!50, inner sep=0pt,font=\bfseries\fontsize{30}{35}\selectfont}]
            \begin{scope}[yscale=0.601]
            \input{pic/370.Danbube_river/ours2_1_0.005.tex}
            \end{scope}
        \end{tikzpicture}\hspace{1.5em}
        \begin{tikzpicture}[scale=0.1471,transform shape,every node/.style={circle, draw, outer sep=0pt,minimum size=35pt,fill=danubeBlue!50, inner sep=0pt,font=\bfseries\fontsize{30}{35}\selectfont}]
            \begin{scope}[yscale=0.601]
            \input{pic/370.Danbube_river/ours2_1_0.01.tex}
            \end{scope}
        \end{tikzpicture}\hspace{1.5em}
        \begin{tikzpicture}[scale=0.1471,transform shape,every node/.style={circle, draw, outer sep=0pt,minimum size=35pt,fill=danubeBlue!50, inner sep=0pt,font=\bfseries\fontsize{30}{35}\selectfont}]
            \begin{scope}[yscale=0.601]
            \input{pic/370.Danbube_river/ours2_1_0.05.tex}
            \end{scope}
        \end{tikzpicture}
    \end{subfigure}\\[0.8em]
    \begin{subfigure}[b]{\textwidth}
        \centering
        \begin{tikzpicture}[scale=0.1471,transform shape,every node/.style={circle, draw, outer sep=0pt,minimum size=35pt,fill=danubeBlue!50, inner sep=0pt,font=\bfseries\fontsize{30}{35}\selectfont}]
            \begin{scope}[yscale=0.601]
            \input{pic/370.Danbube_river/ours2_5_0.001.tex}
            \end{scope}
        \end{tikzpicture}\hspace{1.5em}
        \begin{tikzpicture}[scale=0.1471,transform shape,every node/.style={circle, draw, outer sep=0pt,minimum size=35pt,fill=danubeBlue!50, inner sep=0pt,font=\bfseries\fontsize{30}{35}\selectfont}]
            \begin{scope}[yscale=0.601]
            \input{pic/370.Danbube_river/ours2_5_0.005.tex}
            \end{scope}
        \end{tikzpicture}\hspace{1.5em}
        \begin{tikzpicture}[scale=0.1471,transform shape,every node/.style={circle, draw, outer sep=0pt,minimum size=35pt,fill=danubeBlue!50, inner sep=0pt,font=\bfseries\fontsize{30}{35}\selectfont}]
            \begin{scope}[yscale=0.601]
            \input{pic/370.Danbube_river/ours2_5_0.01.tex}
            \end{scope}
        \end{tikzpicture}\hspace{1.5em}
        \begin{tikzpicture}[scale=0.1471,transform shape,every node/.style={circle, draw, outer sep=0pt,minimum size=35pt,fill=danubeBlue!50, inner sep=0pt,font=\bfseries\fontsize{30}{35}\selectfont}]
            \begin{scope}[yscale=0.601]
            \input{pic/370.Danbube_river/ours2_5_0.05.tex}
            \end{scope}
        \end{tikzpicture}
    \end{subfigure}\\[0.8em]
    \begin{subfigure}[b]{\textwidth}
        \centering
        \begin{tikzpicture}[scale=0.1471,transform shape,every node/.style={circle, draw, outer sep=0pt,minimum size=35pt,fill=danubeBlue!50, inner sep=0pt,font=\bfseries\fontsize{30}{35}\selectfont}]
            \begin{scope}[yscale=0.601]
            \input{pic/370.Danbube_river/ours2_10_0.001.tex}
            \end{scope}
        \end{tikzpicture}\hspace{1.5em}
        \begin{tikzpicture}[scale=0.1471,transform shape,every node/.style={circle, draw, outer sep=0pt,minimum size=35pt,fill=danubeBlue!50, inner sep=0pt,font=\bfseries\fontsize{30}{35}\selectfont}]
            \begin{scope}[yscale=0.601]
            \input{pic/370.Danbube_river/ours2_10_0.005.tex}
            \end{scope}
        \end{tikzpicture}\hspace{1.5em}
        \begin{tikzpicture}[scale=0.1471,transform shape,every node/.style={circle, draw, outer sep=0pt,minimum size=35pt,fill=danubeBlue!50, inner sep=0pt,font=\bfseries\fontsize{30}{35}\selectfont}]
            \begin{scope}[yscale=0.601]
            \input{pic/370.Danbube_river/ours2_10_0.01.tex}
            \end{scope}
        \end{tikzpicture}\hspace{1.5em}
        \begin{tikzpicture}[scale=0.1471,transform shape,every node/.style={circle, draw, outer sep=0pt,minimum size=35pt,fill=danubeBlue!50, inner sep=0pt,font=\bfseries\fontsize{30}{35}\selectfont}]
            \begin{scope}[yscale=0.601]
            \input{pic/370.Danbube_river/ours2_10_0.05.tex}
            \end{scope}
        \end{tikzpicture}
    \end{subfigure}
    \caption{Estimated graphs for Danube river dataset with $\tau = 2$ across different quantile level $q$ (rows: $0.995$, $0.99$, $0.95$, $0.9$) and $\alpha$ (columns: $0.001$, $0.005$, $0.01$, $0.05$).}
    \label{fig:danubeSensitivityTau2}
\end{figure}

\end{appendices}

\end{document}